\documentclass[11pt,letterpaper]{article}
\usepackage{fullpage}
\usepackage{amsmath,amssymb,amsthm,mathtools}
\usepackage{mathpazo}
\usepackage[T1]{fontenc}
\usepackage[full]{textcomp}
\usepackage[american]{babel}
\usepackage[varg]{pxfonts}
\usepackage{tgpagella}
\usepackage{lmodern}
\usepackage{enumitem}
\usepackage{booktabs}
\usepackage{array}
\usepackage[dvipsnames]{xcolor}
\usepackage{tikz}
\usepackage{placeins}
\usetikzlibrary{arrows.meta}
\usepackage{hyperref}
\usepackage[normalem]{ulem}
\usepackage{xcolor}

\newcommand{\doi}[1]{doi: \href{https://doi.org/#1}{\nolinkurl{#1}}}
\hypersetup{colorlinks=true,linkcolor=blue,citecolor=OliveGreen,urlcolor=blue,pdftitle={Cubical Sheaf Complexes with Constant Expansion with Applications to Asymptotically Good qLTCs},pdfauthor={Anonymous}}
\allowdisplaybreaks
\renewcommand{\leq}{\leqslant}
\renewcommand{\le}{\leqslant}

\renewcommand{\ge}{\geqslant}
\let\epsilon\varepsilon
\newtheorem{theorem}{Theorem}[section]
\newtheorem{lemma}[theorem]{Lemma}
\newtheorem{proposition}[theorem]{Proposition}
\newtheorem{corollary}[theorem]{Corollary}
\theoremstyle{definition}
\newtheorem{definition}[theorem]{Definition}
\newtheorem{example}[theorem]{Example}
\theoremstyle{remark}
\newtheorem{remark}[theorem]{Remark}
\newcommand{\F}{\mathbb F}
\newcommand{\PP}{\mathbb P}
\newcommand{\cO}{\mathcal O}
\newcommand{\cB}{\mathcal B}
\newcommand{\cF}{\mathcal F}
\newcommand{\X}{\mathfrak{X}}
\newcommand{\im}{\operatorname{im}}
\newcommand{\supp}{\operatorname{supp}}
\newcommand{\rank}{\operatorname{rank}}
\newcommand{\RS}{\operatorname{RS}}
\newcommand{\Nrd}{\operatorname{Nrd}}
\newcommand{\SL}{\operatorname{SL}}
\newcommand{\PGL}{\operatorname{PGL}}
\newcommand{\PSL}{\operatorname{PSL}}
\newcommand{\CSS}{\operatorname{CSS}}
\newcommand{\dist}{\operatorname{dist}}
\newcommand{\diag}{\operatorname{diag}}
\newcommand{\Hom}{\operatorname{Hom}}
\newcommand{\one}{\mathbf 1}
\newcommand{\db}{d_b}
\newcommand{\coloc}{\operatorname{coloc}}
\newcommand{\eps}{\varepsilon}
\newcommand{\minus}{\setminus}

\makeatletter
\let\rs@mark\relax
\def\rs@markonly{\rs@mark}
\newcommand{\rs@rule}{\textcolor{red}{\rule[0.5ex]{2pt}{0.4pt}}}
\newcommand{\rs@sout}[1]{\bgroup\markoverwith{\rs@rule}\ULon{#1}}
\newcommand{\rs@delsout}[1]{\textcolor{red}{\rs@sout{#1}}}
\long\def\rs@chunk#1{\def\rs@ctmp{#1}%
  \ifx\rs@ctmp\@empty\let\rs@cnext\relax
  \else\long\def\rs@cnext{\rs@style{#1}}\fi
  \rs@cnext}
\long\def\rs@paras#1{\rs@parloop#1\par\rs@mark\rs@stop}
\long\def\rs@parloop#1\par#2\rs@stop{%
  \rs@chunk{#1}%
  \def\rs@ptmp{#2}%
  \ifx\rs@ptmp\rs@markonly\let\rs@pnext\relax
  \else\long\def\rs@pnext{\par\rs@parloop#2\rs@stop}\fi
  \rs@pnext}
\long\def\rs@strip#1\rs@mark#2\rs@stop{\rs@paras{#1}}
\long\def\rs@disploop#1\[#2\]#3\rs@stop{%
  \rs@strip#1\rs@mark\rs@stop
  \def\rs@dtmp{#3}%
  \ifx\rs@dtmp\@empty\let\rs@dnext\relax
  \else\long\def\rs@dnext{\[\color{red}#2\]\rs@disploop#3\rs@stop}\fi
  \rs@dnext}
\long\def\rs@run#1{\begingroup\rs@disploop#1\rs@mark\[\]\rs@stop\endgroup}
\long\def\redsout#1{\let\rs@style\rs@sout\rs@run{#1}}
\long\def\del#1{\let\rs@style\rs@delsout\rs@run{#1}}
\makeatother

\title{Cubical Sheaf Complexes with Constant Expansion\\[2mm]
{\Large with Applications to Asymptotically Good qLTCs}}

\author{
  \setlength{\tabcolsep}{4pt}
  \begin{tabular}{@{}cccc@{}}
    Yeyuan Chen &
    Miryam Mi-Ying Huang &
    Yinchen Liu &
    Er-Cheng Tang\footnote{The authorship is in alphabetical order.}\\[4pt]
    {\small University of Michigan, Ann Arbor} &
    {\small Carnegie Mellon University} &
    {\small Tsinghua University} &
    {\small University of Washington}
  \end{tabular}
}
\date{}
\usepackage{tcolorbox}

\newtcolorbox{aidisclosure}{
  colback=white,
  colframe=black,
  boxrule=0.4pt,
  sharp corners,
  boxsep=0pt,
  left=6pt,
  right=6pt,
  top=6pt,
  bottom=6pt,
  fontupper=\small,
  before skip=24pt,
  after skip=12pt,
  before upper={%
    \setlength{\parindent}{0pt}%
    \setlength{\parskip}{0.5\baselineskip}%
    \textbf{AI disclosure and Acknowledgement.}\enspace
  }
}

\begin{document}
\pagenumbering{gobble}
\hypersetup{pageanchor=false}

\maketitle

\begin{abstract}
For every fixed integers $r \ge 4$ and $2 \le k \le r-2$, we construct $r$-dimensional cubical sheaf complexes whose degree-$k$ CSS codes have positive constant rate, linear distance, and constant soundness, with bounded row and column weights. Taking $r=4$ and $k=2$ gives a family of asymptotically good binary qLTCs.

At the core of our construction is a uniform product-expansion theorem for explicit Reed-Solomon codes on norm-one evaluation sets. The key point is that the expansion constant stays bounded away from zero as the local code lengths grow.  We place these codes on arithmetic cubical complexes, obtaining constant local expansion for both the resulting sheaf and its dual. Together with the local-to-global framework of Dinur, Lin, and Vidick (FOCS 2024) and sheaf duality, this gives linear distance and constant soundness, while an asymmetric choice of local code dimensions gives positive rate. The resulting codes are explicit and polynomial-time computable.

\end{abstract}
\vspace{1.5em}
\begin{aidisclosure}
Some of the work ruling out edge cases came from earlier models (GPT-5.6 Pro and Opus 5), while the main construction
and technical ideas in the current paper were generated with GPT-6 Astra on September 7th
\href{https://drive.google.com/drive/folders/1dV3hYreuELOS7l-_9shsAW9x-OEkjDVF?usp=sharing}{(See here for original GPT generated write-up)}. The original GPT-generated \href{https://drive.google.com/file/d/1zdFPhx4UjV6OiBzXrXj1xPEqAZ5umGFg/view?usp=share_link}{write-up} was hard to read and used many
ad hoc terms, including terminology that is not standard in coding
theory. We believe that this result is important and that human
understanding of the proof is essential. In this era of AI-assisted research, our focus is less on
personal recognition or being first, and more on advancing
human knowledge in a form that others can understand, build
on, and apply. We therefore spent
substantial time internalizing, checking and rewriting the proofs, and revising
the paper. Our goal was to make the arguments clear to coding
theorists and to give the broader theoretical computer science
community some intuition behind the construction.

Parts of the appendices are based on material from an unpublished
project by Yinchen Liu and Ryan O'Donnell. We thank Andrea Coladangelo
for his useful feedbacks.

We compare this result with the recent works of Bafna, Li, and Nguyen,
and of Gay and Jeronimo, in the related work section \ref{related-work}.

\end{aidisclosure}

\clearpage
\begingroup
\small
\tableofcontents
\clearpage
\endgroup
\pagenumbering{arabic}
\hypersetup{pageanchor=true}

\section{Introduction} \label{sec:introduction}
Quantum error correction asks how to store quantum information so that
it survives noise without too much redundancy. A quantum code with
parameters $[[N,K,D]]$ encodes $K$ logical qubits into $N$ physical
qubits and can correct any error on fewer than $D/2$ qubits. Its
\emph{rate} is $K/N$, and an asymptotically good family has
$K=\Omega(N)$ and $D=\Omega(N)$. The code should also admit simple \emph{local checks}.
In a quantum low-density parity-check (qLDPC) code, every check involves
only a constant number of qubits, and every qubit participates in only
a constant number of checks.

However, qLDPC codes are not the whole story. Even with sparse checks, a corrupted word can be very far from every valid codeword while violating only a small fraction of the
checks. Quantum locally testable codes (qLTCs), introduced by
Aharonov and Eldar~\cite{AharonovEldar2015}, are codes that rule out this undesirable behavior by an additional locality condition called \emph{local testability}. 
More formally, a qLDPC code $\CSS(C_X,C_Z)$ is locally testable with soundness \(\rho\) if for every constituent code $C \in\{ C_X, C_Z\}$ and every word \(u\), the probability that a uniformly sampled parity check of $C$ rejects \(u\) is at least $\rho\cdot {\operatorname{dist}(u,C)}/{N}$.
In other words, a word at
relative distance $\delta$ violates at least a $\rho\delta$ fraction
of the checks. We write $[[N,K,D,\rho]]$ for a qLTC with these four
parameters, and extend the notation to a family by asking that the
stated bounds hold for every member.

The natural goal is to achieve all of these properties simultaneously: a
qLDPC code
with constant rate, linear distance, and constant soundness, with all
constants independent of $N$. This asks for more than an asymptotically good qLDPC code:
while the qLDPC condition guarantees that the constraints remain sparse,
local testability further requires these sparse constraints to faithfully
reflect global distance from the code, so that globally significant
errors remain visible through local checks.  There has been substantial
progress toward this goal~\cite{PanteleevKalachev2022,LeverrierZemor2022,
QuantumCodesFromHighDimensionalManifolds,
TowardsLocalTestabilityForQuantumCoding,CrossEtAl2024,
TradeoffConstructionsForQLTCs}.  Most closely related to our work,
Dinur, Lin, and Vidick~\cite{DinurLinVidick2024} constructed qLTCs with
constant rate and bounded check weights, while incurring only
polylogarithmic losses in relative distance and soundness. Their result
therefore brings the problem close to the desired regime, leaving the
following natural question:

\begin{quote}
\centering {\it Can qLTCs simultaneously achieve constant rate, linear distance, constant soundness, and bounded check and qubit degrees? \footnote{Here, of course, all constants must be independent of the block length.}}
\end{quote}

\subsection{Main results}
We construct such codes using cubical sheaf complexes. The dimension
$r$ and coding degree $k$ can be fixed independently within the range
$2\le k\le r-2$.
\begin{theorem}[Main result, informal]
\label{intro}
For every fixed $r\ge4$ and $2\le k\le r-2$, we construct an explicit,
polynomial-time computable family of $r$-dimensional cubical sheaf
complexes whose degree-$k$ CSS codes have block length $N\to\infty$ and
$$
K=\Omega(N),\qquad
D=\Omega(N),\qquad
\rho=\Omega(1).$$
The associated check matrices have uniformly bounded row and column
weights. These guarantees hold on both the primal and dual sides and
are preserved after binary expansion.
\end{theorem}

Here the constants may depend on $(r,k)$ but not on $N$. By explicit,
we mean that the check matrices can be computed in time polynomial
in $N$.

Asymptotically good qLDPC codes were constructed by Panteleev and
Kalachev~\cite{PanteleevKalachev2022} and Leverrier and
Z{\'e}mor~\cite{LeverrierZemor2022}. These give constant rate, linear
distance, and bounded check weights, but do not provide constant
soundness. Earlier qLTC constructions gave different tradeoffs among
rate, distance, soundness, and check weight
\cite{QuantumCodesFromHighDimensionalManifolds,
TowardsLocalTestabilityForQuantumCoding,CrossEtAl2024,
TradeoffConstructionsForQLTCs}.

Most closely related to our work, Dinur, Lin, and
Vidick~\cite{DinurLinVidick2024} constructed qLTCs with constant rate
and constant check weight, with distance
$D=\Omega(N/\log^3 N)$ and soundness $\rho=\Omega(1/\log^3 N)$.
Our construction removes these polylogarithmic losses while keeping
both the check weight and the qubit degree bounded. We discuss prior
and concurrent works further in Section~\ref{related-work}.

Our construction combines arithmetic quotients of products of
trees~\cite{RSV2019,HLMRZ2025} with compatible Reed-Solomon local
codes. Stronger geometric expansion requires larger tree degrees,
which also makes the local codes longer. We prove that these codes
have product expansion with a constant that stays uniform as their
lengths grow. This gives uniform local expansion for the sheaf and
its dual. The local-to-global framework of Dinur, Lin, and
Vidick~\cite{DinurLinVidick2024} then gives linear distance and
constant soundness, while a choice of different local code dimensions
in different directions gives positive rate.

Taking $r=4$ and $k=2$ gives the following binary application.

\begin{theorem}[Four-dimensional binary application]\label{thm:main-original}
There are absolute constants $\alpha,\delta,\rho>0$, an integer
$\Delta$, and a polynomial-time computable family of binary matrices
$H_X,H_Z$ with $N\to\infty$ columns, such that
\[
 H_XH_Z^{\mathsf T}=0,\qquad
 \CSS(H_X,H_Z)=[[N,K,D,\rho]],\qquad
 K\ge\alpha N,\quad D\ge\delta N.
\]
Every row and column of either matrix has weight at most $\Delta$.
Writing $m_X,m_Z$ for their numbers of rows, we have $m_X,m_Z=\Theta(N)$,
and for every $x,z\in\F_2^{N}$,
\begin{align*}
 \frac{|H_Xx|}{m_X}&\ge\rho\frac{\dist(x,\ker H_X)}{N},
 &\frac{|H_Zz|}{m_Z}&\ge\rho\frac{\dist(z,\ker H_Z)}{N}.
 \label{eq:main-soundness}
\end{align*}
\end{theorem}

\subsection{Technical overview}
\label{sec:technical-overview}
\label{sec:technical-overview-new}

Our construction uses decorated cubical complexes coming from arithmetic quotients of products of trees \cite{RSV2019,HLMRZ2025}. At a high level, we use the tree directions to define local tensor codes, descend these codes to a finite cubical quotient, and then apply the local-to-global framework of Dinur, Lin, and Vidick~\cite{DinurLinVidick2024} to obtain a quantum code. 

We first briefly recall the language of cubical complexes and Tanner sheaves. A cubical complex consists of vertices, edges, squares, and higher-dimensional cubes glued together along their faces. A sheaf assigns a code space $\cF(\sigma)$ to each face $\sigma$, together with restriction maps between incident faces. These maps assemble into a cochain complex
\[
 C^j(\X;\cF)=\bigoplus_{\sigma\in\X(j)}\cF(\sigma),
 \qquad \delta^{j+1}\delta^j=0.
\]
Writing $B^j$ for the matrix of $\delta^j$ in binary coordinates, we take
\[
 H_X=B^k,\qquad H_Z=(B^{k-1})^{\mathsf T},\qquad
 N=\dim_{\F_2}C^k,\qquad K=\dim_{\F_2}H^k.
\]
The identity $\delta^k\delta^{k-1}=0$ gives the CSS condition $H_XH_Z^{\mathsf T}=0$, which is exactly the CSS commutation condition. Moreover, each check only involves nearby incident faces. Thus, if the cubical complex has bounded degree, then the resulting quantum code has bounded check weight and bounded qubit degree. We will also use the dual Tanner sheaf $\cF^\perp$, obtained by replacing the generating local codes by their duals, to control the other CSS side.

For this construction to have constant relative distance and soundness, we need two main conditions. We state these conditions below using the framework and terminology of Dinur, Lin, and Vidick~\cite{DinurLinVidick2024} and the quantitative sheaf-duality formulation of~\cite{LiEtAl2025}.

\begin{proposition}[Local-to-global conditions, informal]
\label{overview:criterion}
Fix $r\ge4$ and $\kappa>0$. For the cubical complexes constructed here
there is $\lambda_0=\lambda_0(r,\kappa)>0$ such that the following
holds. Suppose that
\begin{enumerate}[label=(\roman*),leftmargin=*]
\item the parallel face graphs of $\X$ satisfy a spectral bound with
parameter $\lambda\le\lambda_0$, and
\item the local complexes of both $\cF$ and $\cF^\perp$ have coboundary
expansion at least $\kappa$.
\end{enumerate}
Then, for $2\le k\le r-2$, the degree-$k$ CSS code has bounded check
weights and parameters
\[
 [[N,K,D,\rho]],\qquad N=\Theta(|\X(0)|),\qquad
 D=\Omega(N),\qquad \rho=\Omega(1).
\]
Positive rate $K=\Omega(N)$ is arranged separately, in
Section~\ref{sec:rate}.
\end{proposition}

Let us briefly explain the two conditions. Condition~(i) is the geometric one. For each family of parallel faces of $\X$, we consider the graph that connects faces that meet in a common higher-dimensional cube. We require these graphs to be connected spectral expanders, with expansion parameter $\lambda$ sufficiently small. Intuitively, this means that a small collection of faces cannot have too many connections only among themselves.

Condition~(ii) is the local coding one. For each face $\sigma$, consider its \emph{upper star}, namely the collection of cubes containing $\sigma$. The restriction maps of the sheaf on this neighborhood form a small local complex. Coboundary expansion asks that if a local syndrome is nonzero, then correcting it requires changing a number of local blocks proportional to its distance from a valid local configuration. We require this expansion to be at least a fixed constant $\kappa$, uniformly over all faces and for both $\cF$ and $\cF^\perp$.

The intuition behind the two conditions is that local expansion forces the support of a
minimum-weight nontrivial logical operator to meet many neighboring
cubes, while spectral mixing prevents a small set from having that many
internal neighbors. Together, these two effects rule out nontrivial
logical operators with small support, and give linear distance and
constant soundness. Positive rate is arranged
separately, by giving $k$ directions small fixed local rates and the
remaining ones the complementary rates, as in DLV's rate
argument~\cite[Section~4]{DinurLinVidick2024}.

\paragraph{Geometry.}
The first condition explains why the bounded-degree construction of Dinur, Lin, and Vidick~\cite{DinurLinVidick2024} does not already give good qLTCs. Earlier good qLDPC constructions use two directions coming from left and right multiplication in a non-Abelian group~\cite{PanteleevKalachev2022,LeverrierZemor2022}. These commute because $a(gb)=(ag)b$, but this does not directly give the four directions needed here. Using an Abelian group would make all directions commute, but constant spectral expansion on $M$ vertices would then require $\Omega(\log M)$ generators. Dinur, Lin, and Vidick keep the degrees bounded by using Abelian lifts whose directional graphs split into many disjoint expanding components. Each component expands well internally, but a whole component has no edges leaving it. As the complex grows, one component occupies a smaller fraction of the graph. Thus, the available expansion bounds only apply to sets of inverse-polylogarithmic density, and this loss propagates to the resulting distance and soundness bounds.

We avoid this loss by using arithmetic quotients of products of Bruhat-Tits trees~\cite{RSV2019,HLMRZ2025}. Our complexes have the form
$$
\X_s=\Gamma_s\backslash
\bigl(\mathcal T_1\times\cdots\times\mathcal T_r\bigr),
$$
where each $\mathcal T_i$ is a $(q_i+1)$-regular tree. We color the vertices of each tree by $0$ and $1$, so that every edge joins vertices of different colors. The group $\Gamma_s$ acts coordinatewise, preserving the tree directions and these colors, and identifies vertices and cubes to give a finite quotient. These quotients are decorated Cayley cubical complexes~\cite{HLMRZ2025}. Their vertices can be written as $(g,x)\in G\times\F_2^r$, where $x_i\in\{0,1\}$ records the vertex color in tree direction $i$. In direction $i$, the neighbors of $(g,x)$ are
\[
 (gb,x+e_i),\qquad b\in B_i,
\]
where $B_i$ is an inverse-closed set of size $q_i+1$. The sets $B_i$ commute setwise, so $B_iB_j=B_jB_i$, with unique factorization for products in distinct directions. Thus, choosing one incident edge in each of several directions determines a unique cube. The individual generators need not commute. As shown in Figure~\ref{fig:edges-and-cubes}, the opposite edges of a square may have different labels.

\begin{figure}[htbp]
\centering
\begin{minipage}[b]{0.49\textwidth}
\centering
\begin{tikzpicture}[
 font=\footnotesize,
 vertex/.style={circle,fill=black,inner sep=1.5pt},
 idir/.style={NavyBlue,thick,-{Stealth[length=4pt]}},
 jdir/.style={BrickRed,thick,-{Stealth[length=4pt]}}
]
 \path[use as bounding box] (-1.4,-0.8) rectangle (5.2,4.05);
 \coordinate (v00) at (0,0.35);
 \coordinate (v10) at (3.8,0.35);
 \coordinate (v01) at (0,3.05);
 \coordinate (v11) at (3.8,3.05);
 \fill[NavyBlue!4] (v00)--(v10)--(v11)--(v01)--cycle;
 \draw[idir] (v00)--node[below] {$b\in B_i$} (v10);
 \draw[idir] (v01)--node[below] {$b'\in B_i$} (v11);
 \draw[jdir] (v00)--node[left] {$c\in B_j$} (v01);
 \draw[jdir] (v10)--node[right] {$c'\in B_j$} (v11);
 \foreach \v in {v00,v10,v01,v11} \node[vertex] at (\v) {};
 \node[below=15pt] at (v00) {$(g,x)$};
 \node[below=15pt] at (v10) {$(gb,x+e_i)$};
 \node[above=6pt] at (v01) {$(gc,x+e_j)$};
 \node[above=6pt] at (v11) {$(gbc',x+e_i+e_j)$};
 \node at (1.9,1.7) {$bc'=cb'$};
\end{tikzpicture}

\small (a) Unique square completion.
\end{minipage}\hfill
\begin{minipage}[b]{0.49\textwidth}
\centering
\begin{tikzpicture}[
 font=\footnotesize,
 fiber/.style={circle,draw=black!35,fill=white,minimum size=19pt,
               inner sep=0pt},
 idir/.style={NavyBlue,thick},
 jdir/.style={BrickRed,thick},
 kdir/.style={ForestGreen,thick}
]
 \path[use as bounding box] (-0.7,-0.8) rectangle (4.7,4.05);
 \coordinate (p000) at (0,0);
 \coordinate (p100) at (2.65,0);
 \coordinate (p010) at (0,2.55);
 \coordinate (p110) at (2.65,2.55);
 \coordinate (p001) at (1.25,0.95);
 \coordinate (p101) at (3.9,0.95);
 \coordinate (p011) at (1.25,3.5);
 \coordinate (p111) at (3.9,3.5);
 \foreach \bits in {000,100,010,110,001,101,011,111}
   \coordinate (v\bits) at ([xshift=5pt,yshift=-3pt]p\bits);
 \fill[NavyBlue!7] (v100)--(v110)--(v111)--(v101)--cycle;
 \foreach \bits in {000,100,010,110,001,101,011,111} {
   \node[fiber] at (p\bits) {};
   \node[font=\scriptsize] at ([yshift=2pt]p\bits) {$G$};
 }
 \foreach \bits in {000,100,001,101}
   \node[below=10pt,font=\scriptsize] at (p\bits) {$\bits$};
 \foreach \bits in {010,110,011,111}
   \node[above=10pt,font=\scriptsize] at (p\bits) {$\bits$};
 \draw[idir,densely dashed] (v001)--(v101);
 \draw[jdir,densely dashed] (v001)--(v011);
 \draw[kdir,densely dashed] (v000)--(v001);
 \draw[idir] (v000)--node[below=2pt] {$B_i$} (v100);
 \draw[idir] (v010)--(v110) (v011)--(v111);
 \draw[jdir] (v000)--node[left=2pt] {$B_j$} (v010);
 \draw[jdir] (v100)--(v110) (v101)--(v111);
 \draw[kdir] (v100)--node[below=2pt] {$B_k$} (v101);
 \draw[kdir] (v010)--(v011) (v110)--(v111);
 \foreach \bits in {000,100,010,110,001,101,011,111}
   \fill (v\bits) circle (1.3pt);
\end{tikzpicture}

\small (b) A cube in eight color classes.
\end{minipage}
\caption{Edges and cubes in a decorated Cayley cubical complex.
In (a), the two paths to the opposite vertex agree by $bc'=cb'$.
The two blue opposite edges are vertices joined by an edge in the
direction-$j$ parallel-face graph.
In (b), the three bits record colors in directions $(i,j,k)$;
each circle represents a copy of $G$, and the dots and segments
select one cube. The shaded square has type $\{j,k\}$.}
\label{fig:edges-and-cubes}
\end{figure}
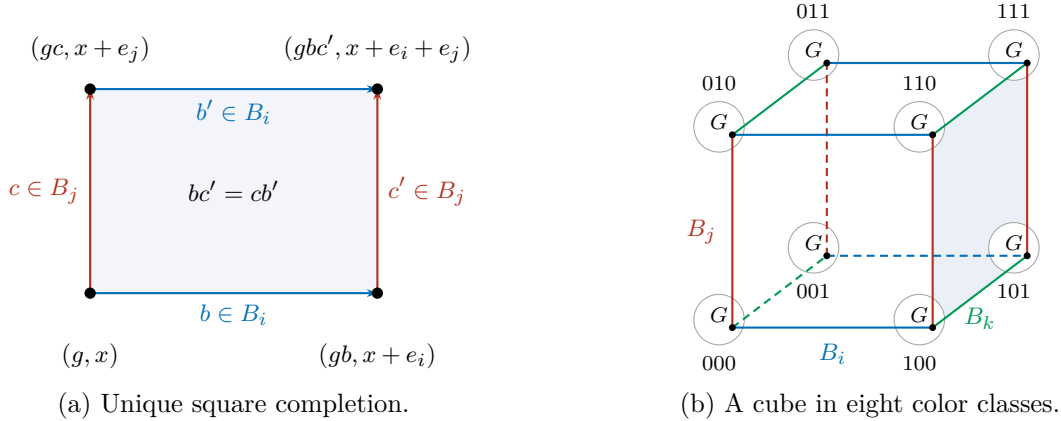

For our arithmetic quotients, the parallel face graphs needed for condition~(i) are connected and satisfy the desired spectral bounds
\[
 \lambda\le\max_i\frac{2\sqrt{q_i}}{q_i+1},
\]
independently of $s$. Therefore, their expansion improves as the tree degrees increase. Once these degrees are fixed, the quotient can grow while its local neighborhoods remain unchanged. 

At the same time, the product-of-trees structure allows us
to build a sheaf code by taking tensor products of sheaf
codes on the individual trees~\cite{PanteleevKalachev2024}.
We first describe its local structure on the product of trees.
At a vertex $v=(v_1,\ldots,v_r)$, let $E_i(v_i)$ be the
set of edges incident to $v_i$ in the $i$th tree, and choose
a code $C_i(v_i)\subseteq F^{E_i(v_i)}$.
Choosing one edge in each direction determines one
$r$-dimensional cube, so the cubes around $v$ form the
grid $\prod_i E_i(v_i)$.
We require every coordinate line in direction $i$ to
carry a word in $C_i(v_i)$.
For example, in two directions, these are simply the
code constraints on the rows and columns of a matrix.
Together, the directional constraints define the
local tensor code
$\mathcal G(v)=\bigotimes_{i=1}^r C_i(v_i)$.

The corresponding local complex is
\[
 \bigotimes_{i=1}^r
 [C_i(v_i)\hookrightarrow F^{E_i(v_i)}].
\]
Each two-term complex describes one directional code
and its inclusion into the space of edge values.
Their tensor product records how these constraints
fit together across the different directions.
To obtain local coboundary expansion, we need these
constraints to work robustly together.

The property we use to ensure this is \emph{product expansion},
a robustness property of the whole tuple of local codes.
Existing results show how product expansion gives
local coboundary expansion for these tensor-product
complexes~\cite{KalachevPanteleev2025,DinurLinVidick2024}.
Around a higher-dimensional face, the local complex
uses only the directions not occupied by that face.
We therefore require product expansion for every
nonempty subtuple of the local codes.
To satisfy condition~(ii), these expansion bounds must
remain bounded away from zero as the local code lengths grow.

\paragraph{Local codes.}
This tells us what expansion property the local codes should have, but we still need to place them consistently on the arithmetic quotient. We cannot simply put arbitrary robust codes on these local grids. After taking the arithmetic quotient, different neighborhoods are identified through the group action, and their local coordinates may be permuted or rescaled. The code constraints must still agree under these identifications. Random codes can give the robustness we want~\cite{KalachevPanteleev2025}, but they do not naturally have this extra structure. Thus, we need local codes that are both robust and compatible with the arithmetic quotient.

Reed-Solomon codes are a natural choice for this purpose. The $q+1$ edges around a vertex of a Bruhat-Tits tree are naturally indexed by the projective line $\PP^1(\F_q)$, namely the set of one-dimensional directions in $\F_q^2$. We can therefore define the local code by evaluating homogeneous polynomials on these $q+1$ directions.

To build the sheaf code on the tree, we start with the local codes at the two endpoints of one standard edge and propagate them using the group action. The group used here preserves the two vertex colors and can send any edge to any other edge. Different group elements may send the standard edge to the same edge. For their copies to fit together, an element fixing this edge must rescale its code coordinate by the same factor at both endpoints.

For homogeneous polynomials of degree $d$, the usual actions give factors $a^{-d}$ and $a^d$ at the endpoints of colors $0$ and $1$, respectively, where $a\in\F_q^\times$ depends on the group element. These factors need not agree. We make them agree by a Frobenius twist at the vertices of color $0$. A group element acts on the local code by a linear substitution in the two variables $X$ and $Y$. Writing $q=Q^2$, the twist raises the coefficients of this substitution to their $Q$th power before substituting. This is again a group action, since raising to the $Q$th power is a field automorphism of $\F_q$. It changes the first factor to $a^{-Qd}$, and taking $d$ to be a multiple of $Q-1$ gives
\[
 (Q-1)\mid d\quad\Longrightarrow\quad
 a^{-Qd}=a^d\qquad(a\in\F_q^\times).
\]
The two factors now agree, so the copies fit together into a sheaf on the whole tree compatible with the group action. Since the Frobenius twist only permutes the projective directions, the local codes remain Reed-Solomon codes up to permutation and rescaling. Replacing $d$ by $q-1-d$ gives the dual local code and the descent also applies to the dual sheaf.

Putting these pieces together, we can place the Reed-Solomon codes in each tree direction, combine them into the product sheaf $\mathcal G$, and pass them consistently to the finite arithmetic quotient. Descent preserves these local tensor complexes, and hence their local expansion. This settles the compatibility part of condition~(ii), but we still need to prove uniform robustness of these local codes. We do this by proving product expansion for these codes and their duals.

\paragraph{Product expansion.}
The construction above makes the local codes compatible
with the arithmetic quotient.
We still need to prove that these codes and their duals
have the uniform robustness required in condition~(ii).

The difficulty is that codewords from different directions
can cancel when added together.
Consider, for example, a word obtained by adding several
row codewords and column codewords in a two-dimensional grid.
Their values may cancel at many intersections, leaving
a final word with much smaller support.
Product expansion asks whether this same word can be
written using a small total length of coordinate lines. For each coordinate line $\ell$, let $C_\ell$ denote the code supported on that line. We ask that every word $x$ in the span of these line codes admit a decomposition $$x=\sum_\ell a_\ell,\quad a_\ell\in C_\ell,\quad \text{s.t.    }\quad\varrho\sum_{\ell:a_\ell\neq0}|\ell|\le |x|,$$ where $|x|$ is the support size of $x$. Importantly, this decomposition does not have to be originally used to form $x$. Even if many line codewords were used and then cancelled heavily, there should be another decomposition of the same word using only a small total length of coordinate lines. We prove this property for the Reed-Solomon codes used in our construction. After a change of coordinates, the projective evaluation points can be written as the norm-one sets $$q_i=2^{2(A+i)},\qquad\Sigma_i=\left\{z\in\F_{q_i^2}: z^{q_i+1}=1\right\},\qquad i\in[r].$$
    Up to coordinate permutations and nonzero rescalings, our local codes are ordinary Reed-Solomon codes evaluated on the sets $\Sigma_i$. Our main coding theorem shows that, when the local rates lie in $[\beta,1-\beta]$, these codes and their dual codes have product expansion $\varrho=c_r\beta^r>0.$
    The same bound also holds for every nonempty subset of the indices. Most importantly, the product expansion factor $\varrho$ is independent of the choice of $A$, and hence does not deteriorate as the tree degrees grow. This is the uniform robustness required in condition~(ii). 
    
    We next explain the main idea behind the proof. Let $U=\supp(x)$ be the support of the final word. The set $U$ itself may not contain sufficient complete coordinate lines needed for the required decomposition. We therefore enlarge $U$ to a set $Z$ that contains enough complete lines, while keeping $Z$ only a constant factor larger than $U$. Following~\cite{KalachevPanteleev2025}, we want $Z$ to be \emph{inner-generated}: every word in the line-code span whose support lies inside $Z$ can be decomposed using only coordinate lines contained in $Z$. Our choice of $Z$ is a low-degree polynomial hull of $U$. Intuitively, $Z$ is enlarged from $U$ by adding all grid points that are forced by the low-degree polynomial relations already satisfied by $U$. For example, if $U$ contains more points on a coordinate line than the allowed polynomial degree, then every such polynomial vanishing on these points must vanish on the whole line. The polynomial hull therefore naturally fills in coordinate lines that are useful for the decomposition.

    The main algebraic step is to show that the polynomial hull is inner-generated. Here the special choice of the evaluation sets $\Sigma_i$ is crucial: every $z\in\Sigma_i$ satisfies $z^{q_i}=z^{-1}$. This identity allows us to use Frobenius operations to replace any multivariate polynomial by a polynomial with exactly the same zero set, but whose gradient is pointed towards a prescribed direction. The replacement increases the coordinate degrees by only a constant factor. We use this flexibility to construct, via determinants, low-degree polynomials that isolate individual points of a common zero set. Combined with the Reed--Solomon parity constraints, they force the support of a codeword to lie in the union of coordinate lines contained in the polynomial hull. An interpolation and gluing argument then yields a decomposition into codewords supported on these lines, establishing inner generation. The full argument appears in Appendix~\ref{sec:algebra}.

    Once $Z$ is inner-generated, the product-expansion bound follows by counting. Lines in any fixed direction are disjoint, so the total length of all coordinate lines used inside $Z$ is at most $r|Z|$. Hence the resulting decomposition satisfies $\sum_{\ell:a_\ell\neq0}|\ell|\le r|Z|\le rK|x|,$ which gives $\varrho\ge \frac{1}{rK}.$ Since $K$ depends only on $r$ and $\beta$, this bound is uniform in the local code lengths.
    
    Finally, we can combine the two conditions. We first choose the local rates so that both the sheaf and its dual have positive rate and uniform local robustness. We then choose the tree degrees large enough so that the arithmetic quotients satisfy the geometric expansion of condition~(i). After these local parameters are fixed, we grow only the size of the finite quotient. Applying the local-to-global argument of Dinur, Lin, and Vidick together with quantitative sheaf duality~\cite{LiEtAl2025} gives linear distance and constant soundness on both CSS sides. Since the local data remain fixed, passing to binary coefficients preserves these guarantees while keeping the check weights and qubit degrees uniformly bounded. 

\subsection{Related work}
\label{related-work}

The construction of asymptotically good qLDPC codes by Panteleev and
Kalachev~\cite{PanteleevKalachev2022} and Leverrier and
Z{'e}mor~\cite{LeverrierZemor2022} showed that constant rate, linear
distance, and constant-weight checks can be achieved simultaneously.
Additionally, obtaining local testability has proved more difficult.
Earlier and more recent qLTC constructions explored a range of tradeoffs
among rate, distance, soundness, and locality
~\cite{QuantumCodesFromHighDimensionalManifolds,
TowardsLocalTestabilityForQuantumCoding,CrossEtAl2024,
GeneralDistanceBalancingForQLTCs,TradeoffConstructionsForQLTCs}.
Most closely related to our setting, Dinur, Lin, and
Vidick~\cite{DinurLinVidick2024} obtained constant rate and constant
locality with only polylogarithmic losses in relative distance and
soundness.

Our construction also builds on the theory of sheaf codes and local
complexes~\cite{PanteleevKalachev2024}, together with product expansion
for tuples of local codes.  Kalachev and
Panteleev~\cite{KalachevPanteleev2025} proved product expansion for
arbitrary fixed tuples of random codes and introduced the
inner-generation framework that we use in our Reed-Solomon analysis.
Bafna and Vyas~\cite{BafnaVyas2026} studied product expansion for
Reed-Solomon codes on multiplicative subgroups of coprime orders,
proving a two-factor result and formulating a higher-dimensional
conjecture.  Our norm-one evaluation sets satisfy different arithmetic
constraints, so our result is complementary and does not resolve their
conjecture.

For the underlying geometry, we use arithmetic quotients of products
of trees in the function-field framework of~\cite{RSV2019}, together
with the cubical formulation of~\cite{HLMRZ2025}.  These complexes
provide the high-dimensional geometry on which our Reed-Solomon
sheaves are placed.

\paragraph{Comparison with concurrent works.}

Concurrently, Gay and Jeronimo~\cite{GayJeronimo2026} also construct
asymptotically good qLTCs.  Our proof is considerably shorter and more
modular: we invoke
existing local-to-global and sheaf-duality theorems and concentrate the
main new argument on Reed-Solomon product expansion via
inner generation.  This gives a relatively direct route from the local
coding statement to the final qLTC parameters.

Bafna, Li, and Nguyen~\cite{BafnaLiNguyen2026} give a closely
related construction using arithmetic cubical complexes
and Reed-Solomon local codes.
Both constructions use Frobenius twists to make the local
codes compatible with the geometry.
Their qLTC theorem assumes that certain Reed-Solomon code
tuples have product expansion with a constant independent
of their lengths, which is their Conjecture~5.6. In this work, we prove this expansion property for explicit norm-one
families, using polynomial hulls and the inner-generation
theorem proved in Appendix~\ref{sec:algebra}.
In particular, our result (see Theorem~\ref{thm:rs-product-expansion})
covers the parameter family used in their construction,
thereby partially resolving their conjecture.\footnote{
The exponents $(L+8,L+16,L+32,L+64)$ in their
Definition~5.8 correspond to directions $\{1,5,13,29\}$
of our theorem with $r=29$ and $A=L/2+3$.
If necessary, extend the coefficient field to contain all
29 evaluation sets and then descend the line-code
decompositions by a linear projection fixing the original
field.
This preserves code membership and does not increase
support.
}
The constant remains fixed as the local code lengths grow,
so we can improve the geometric mixing without losing
local expansion.
This gives good qLTCs without a product-expansion conjecture.
Our construction also works in every fixed dimension
$r\ge4$ and every prescribed degree $2\le k\le r-2$,
while their qLTC construction is presented in dimension
four with qubits in degree two.

Li, Li, and Liu~\cite{LiLiLiu2026} focus on a different direction. Using their earlier framework for logical gates~\cite{TransversalGatesOnAlmostGoodQuantumCodes}, they show that the good qLTCs of Gay and Jeronimo~\cite{GayJeronimo2026} can also support transversal multi-controlled-$Z$ gates. Their construction therefore gives codes that are simultaneously good qLDPC codes, good qLTCs, and compatible with these transversal logical gates. Our work focuses instead on improving the code parameters, and we do not study logical gates.

\paragraph{Roadmap}
Section~\ref{sec:preliminaries} introduces CSS codes and
local testability, Tanner sheaves and local complexes,
and decorated Cayley cubical complexes.
Section~\ref{sec:local-to-global} proves a local-to-global
expansion criterion using local robustness,
parallel-face mixing, and sheaf duality.
Section~\ref{sec:rs-construction} constructs Tanner sheaves
on products of trees and their finite quotients, and
realizes this construction with Reed-Solomon local codes
on arithmetic quotients.
Section~\ref{sec:product-expansion} proves uniform product
expansion for these codes and derives local coboundary
expansion for the sheaf and its dual.
Section~\ref{sec:realization} establishes positive rate
and selects the parameters to complete the main results.
Appendix~\ref{sec:arithmetic} provides the arithmetic
and spectral arguments.
Appendix~\ref{sec:algebra} proves the inner-generation
theorem for polynomial zero sets used in the
product-expansion proof.

\section{Preliminaries}
\label{sec:preliminaries}

\subsection{Notation}
Write $[m]=\{1,\ldots,m\}$ and let $e_i$ denote the $i$th
coordinate vector in the ambient coordinate space.
We write $\F_q$ for the field with $q$ elements and $\PP^1(F)$
for the projective line, whose points are the one-dimensional
subspaces of $F^2$.
The coefficient field is denoted by $F$ and has characteristic two;
it is finite unless an extension field is explicitly introduced.
All vector spaces, tensor products, and dimensions are over $F$
unless specified otherwise.
For a finite set $S$, $F^S$ is the space of arrays indexed by $S$,
and $\supp x=\{a\in S:x_a\ne0\}$, so $|x|=|\supp x|$.
For $Z\subseteq S$, we identify $F^Z$ with the subspace of $F^S$
supported on $Z$ by extending arrays by zero.
For a code $C\subseteq F^S$, its dual code is
$C^\perp=\{y\in F^S:\langle x,y\rangle=0\text{ for every }x\in C\}$,
where $\langle x,y\rangle=\sum_{a\in S}x_ay_a$.
Its redundancy is $|S|-\dim C$; relative redundancy and relative
minimum distance are obtained by dividing by $|S|$.
For a set $A$, $\one_A$ is its indicator function.

For a specified direct sum $V=\bigoplus_{a\in S}V_a$, each
summand is one \emph{block}, regardless of its dimension. Define
\[
 \supp_b x=\{a\in S:x_a\ne0\},\qquad
 |x|_b=|\supp_b x|,\qquad
 \db(x,U)=\min_{u\in U}|x-u|_b
\]
for $x=(x_a)_{a\in S}$ and a nonempty set $U\subseteq V$.
Ordinary coordinate Hamming
weight and distance are denoted by $|x|$ and $\dist(x,U)$.
Minima and infima over empty sets are $+\infty$.

\subsection{CSS codes and local testability}
\label{sec:css-parameters}
Binary matrices $H_X,H_Z$ with $N$ columns and
$H_XH_Z^{\mathsf T}=0$ define the CSS code
$\CSS(H_X,H_Z)=[[N,K,D]]$, with $N$ physical qubits,
$K$ logical qubits, and minimum distance $D$, where
\begin{align*}
 K&=N-\rank H_X-\rank H_Z,\\
 D&=\min\left\{
   \min_{x\in\ker H_X\minus\im H_Z^{\mathsf T}}|x|,
   \min_{z\in\ker H_Z\minus\im H_X^{\mathsf T}}|z|
   \right\}.
\end{align*}

We use the normalization of~\cite[Definition~2.6]{DinurLinVidick2024}.
\begin{definition}[Quantum local testability]
\label{def:normalized-soundness}
A binary check matrix $H$ with $m>0$ rows and $N>0$ columns has
soundness at least $\rho$ if
\[
 \frac{|Hx|}{m}\ge\rho\,\frac{\dist(x,\ker H)}{N}
 \qquad(x\in\F_2^N).
\]
A CSS code is quantum locally testable with soundness at least $\rho$
if both check matrices have soundness at least $\rho$.
\end{definition}

{The vector $Hx$ is the syndrome, and $|Hx|$ counts the violated
checks. Thus ${|Hx|}/{m}$ is the rejection probability when a row of $H$
is chosen uniformly, while $\dist(x,\ker H)/{N}$ is the minimum fraction
of coordinates that must be changed to satisfy all checks.
For a CSS code, each check matrix is normalized by its own number
of rows.}

As in the introduction, we write $[[N,K,D,\rho]]$ for the parameters of
a CSS code that is quantum locally testable with soundness at least
$\rho$, and we use the same notation for a family, meaning that each
stated bound holds for every member. We seek families with constant
rate $K/N=\Omega(1)$, linear distance $D=\Omega(N)$, uniformly bounded
check weights, and constant soundness $\rho=\Omega(1)$.

\subsection{Sheaf codes and local complexes}
\label{sec:cubical-sheaf}\label{sec:sheaf-definition}
We use sheaf codes in the sense of
\cite[Section~III-B]{PanteleevKalachev2024}. Tanner sheaves, generated
by codes on the codimension-one cells, will be our main example;
their duals follow~\cite[Definitions~3.10--3.11]{LiEtAl2025}.

Let $\X$ be an $r$-dimensional regular cell complex,\footnote{Regularity
means that each characteristic map is a homeomorphism from a closed
ball onto the cell's closure. We also assume that $\X$ is locally
finite (every point has a neighborhood meeting only finitely many
cells) and pure (every cell is a face of an $r$-cell).}
and write $\X(j)$ for its set of $j$-cells.
The relation $\sigma\le\tau$ means that $\sigma$ is a face of $\tau$;
write $\sigma\lessdot\tau$ when additionally
$\dim\tau=\dim\sigma+1$. The upper star of $\sigma$ is
$\X_{\ge\sigma}=\{\tau:\sigma\le\tau\}$, with
$\X_{\ge\sigma}(j)=\X_{\ge\sigma}\cap\X(j)$.
For a cell $\sigma$, let $U_\sigma=\{\omega\in\X(r):\sigma\le\omega\}$
be the set of top-dimensional cells containing $\sigma$.

\begin{definition}[Sheaf codes]
\label{def:sheaf-code}
A \emph{sheaf code} $\cF$ on $\X$ assigns a linear code
$\cF(\sigma)\subseteq F^{U_\sigma}$ to each cell $\sigma$, with
$\cF(\omega)=F$ for every top-dimensional cell $\omega$, such that
\[
 c\in\cF(\sigma),\quad \sigma\le\tau
 \quad\Longrightarrow\quad c|_{U_\tau}\in\cF(\tau).
\]
For $\sigma\le\tau$, coordinate projection defines the restriction
$
 \rho_{\sigma,\tau}:\cF(\sigma)\longrightarrow\cF(\tau),c\longmapsto c|_{U_\tau}.
$
In particular, $\rho_{\sigma,\omega}(c)=c_\omega$ for
$\omega\in U_\sigma$.
These coordinate projections automatically satisfy
\[
 \rho_{\sigma,\sigma}=\mathrm{id},\qquad
 \rho_{\tau,\omega}\circ\rho_{\sigma,\tau}
 =\rho_{\sigma,\omega}\quad(\sigma\le\tau\le\omega).
\]
\end{definition}

In this paper, we are primarily focused on Tanner sheaf codes and their duals.
\begin{example}[Tanner sheaf codes and their duals]
\label{ex:tanner-sheaf}\label{def:local-orthogonal-sheaf}
For each $(r-1)$-cell $\tau$, fix a local code
$\mathcal C_\tau\subseteq F^{U_\tau}$.
The \emph{Tanner sheaf} generated by
$\{\mathcal C_\tau\}_{\tau\in \X(r-1)}$
assigns, for cells $\sigma$ and $\omega$ with $\dim\sigma\le r-1$ and $\omega\in\X(r)$, the spaces
\begin{equation}
 \cF(\sigma):=
 \bigl\{c\in F^{U_\sigma}:c|_{U_\tau}\in\mathcal C_\tau
 \text{ for all }\tau\in\X_{\ge\sigma}(r-1)\bigr\},
 \qquad \cF(\omega):=F.
 \label{eq:sheaf-spaces}
\end{equation}
Specifically, $\cF(\tau)=\mathcal C_\tau$ for every $(r-1)$-cell $\tau$.

The \emph{dual Tanner sheaf} $\cF^\perp$ is generated using the same rule by
$\mathcal C_\tau^\perp\subseteq F^{U_\tau}$, the dual code taken with respect to the standard coordinate
pairing on $F^{U_\tau}$. 
\end{example}

\paragraph{Cochains and chains.}
\label{sec:sheaf-differential}
For any sheaf code $\cF$ on $\X$, the cochain spaces and differential are
\cite[Section~VI]{PanteleevKalachev2024}
\begin{equation*}
 C^k=C^k(\X;\cF)=
       \bigoplus_{\sigma\in \X(k)}\cF(\sigma),\qquad
 (\delta^k x)_\tau=
       \sum_{\substack{\sigma\lessdot\tau\\\dim\sigma=k}}
            \rho_{\sigma,\tau}(x_\sigma).
\end{equation*}
For a regular cell complex $\X$, we have $\delta^{k+1}\delta^k=0$
in characteristic $2$.

Its dual chain complex has $C_k=(C^k)^*$ and
$\partial_k=(\delta^{k-1})^*$, where $V^*=\Hom_F(V,F)$.
The blocks of $C^k$ are the summands $\cF(\sigma)$, and those of
$C_k$ are the summands $\cF(\sigma)^*$.
Terms and maps outside the specified degree range are zero. Write
\[
 H^k(\X;\cF)=\ker\delta^k/\im\delta^{k-1},\qquad
 H_k(\X;\cF)=\ker\partial_k/\im\partial_{k+1}.
\]
We also write $H^k(C)$ and $H_k(C)$ when the complex is understood.

To study the sheaf near a cell $\sigma$, we restrict the cochain complex
to cochains supported on its upper star, following
\cite[Eq.~(3.1) and Definition~3.6]{LiEtAl2025}.

\begin{definition}[Local complex and local acyclicity]
\label{def:upper-star-complex}\label{def:local-acyclicity}
For an $\ell$-face $\sigma$, its local cochain complex on the
upper star has terms
\begin{equation}
 K_\sigma^j(\cF)=
 \bigoplus_{\substack{\tau\ge\sigma\\\dim\tau=\ell+j}}
 \cF(\tau),\qquad 0\le j\le r-\ell.
 \label{eq:upper-star-complex}
\end{equation}
Its differential $\delta_\sigma$ sums the restriction
maps over incidences within the upper star.
{Here $K_\sigma^0(\cF)=\cF(\sigma)$, and we can regard $K_\sigma^j(\cF)$ as
the subspace of $C^{\ell+j}(\X;\cF)$ of cochains vanishing off the
upper star. Since the upper star is
closed under passing to cofaces, restricting $\delta$ to the upper star subspace defines $\delta_\sigma$ and gives $\delta_\sigma^2=0$.}

The sheaf $\cF$ is \emph{locally acyclic} if each $K_\sigma$ is exact
below its top degree: $\ker\delta_\sigma^j=\im\delta_\sigma^{j-1}$
for $0\le j<r-\ell$, where $\im\delta_\sigma^{-1}=0$.
\end{definition}
In the cubical tensor-product setting of \cite{DinurLinVidick2024}, each $\cF(\sigma)$ is a
tensor product of code spaces. Then, the resulting local complex defined above is in fact isomorphic with DLV's local
complex~\cite[Section~5.1 and Lemma~5.2]{DinurLinVidick2024}.

We also note that a locally acyclic sheaf code is in fact the Tanner
sheaf generated by its codimension-one codes
\cite[Theorem~3.3 and Corollary~3.7]{LiEtAl2025}.

\subsection{Cubical complexes}
\label{sec:cubical-complex}\label{sec:geometry-construction}
For instantiation, we use the decorated Cayley cubical complexes of
\cite[Definitions~3.1--3.2]{HLMRZ2025}.

\begin{definition}[Cubical generating sets]
\label{def:cubical-generators}
Let $G$ be a finite group. Nonempty subsets $B_1,\ldots,B_r\subseteq G$
are cubical generating sets if
\begin{equation}
 B_i^{-1}=B_i,\qquad B_iB_j=B_jB_i\quad(i\ne j),\qquad
 |B_1\cdots B_r|=\prod_{i=1}^r|B_i|.
 \label{eq:cubical-generators}
\end{equation}
\end{definition}
{Here $B_iB_j=\{bc:b\in B_i,\ c\in B_j\}$, so
$B_iB_j=B_jB_i$ is an equality of sets and does not require
individual generators to commute.}
The cardinality condition and set-wise commutation imply that for any distinct set of indices $\{i_1,\dots,i_k\} \subseteq [r]$, each element of $B_{i_1}\cdots B_{i_k}$ has a unique factorization $b_{i_1}\cdots b_{i_k}$ with $b_{i_j}\in B_{i_j}$.

\begin{definition}[Decorated Cayley cubical complex]
\label{def:cayley-cubical}\label{def:colored-cubical}
Given cubical generating sets $B_1,\ldots,B_r$, write $|B_i|=q_i+1$.
The complex $\X=\operatorname{Cay}_{\square}(G;B_1,\ldots,B_r)$
is defined as follows.
\begin{itemize}
\item Its vertex set is $\X(0)=G\times\F_2^r$; the color of
$(g,x)$ is $x$.
\item Its $r$-faces are the sets
$f=\{(g_x,x):x\in\F_2^r\}\subseteq\X(0)$ satisfying
\[
 g_x^{-1}g_{x+e_i}\in B_i\qquad(x\in\F_2^r,\ i\in[r]).
\]
\item For $I\subseteq[r]$, an $I$-subcube of $\F_2^r$ is a set
$C=x_0+\operatorname{span}_{\F_2}\{e_i:i\in I\}$.
The faces over $C$ are
\[
 \X(C)=\{f|_C:f\in\X(r)\},\qquad
 f|_C=\{(g_x,x):x\in C\}.
\]
For $\sigma\in\X(C)$, write $\operatorname{type}(\sigma)=I$
and $\dim\sigma=|I|$; subfaces are obtained by further restriction.
Write $\X(I)$ for all faces of type $I$, and
$\X(j)=\bigcup_{|I|=j}\X(I)$ for all $j$-faces.
\end{itemize}
We identify $\X$ with its standard cubical realization,
a finite, pure $r$-dimensional regular cell complex.
\end{definition}
\begin{lemma}[Face uniqueness and coface counts {\cite[Lemma~3.3]{HLMRZ2025}}]
\label{lem:cubical-incidence}
{A nonempty vertex set lies in at most one face over the smallest
coordinate subcube containing its colors. For a face $\sigma$ of
type $I$ and any $I\subseteq J\subseteq[r]$, we have
\[
 \bigl|\{\tau\ge\sigma:\operatorname{type}(\tau)=J\}\bigr|
 =\prod_{j\in J\setminus I}(q_j+1).
\]
In particular, $|U_\sigma|=\prod_{j\notin I}(q_j+1)$.}
\end{lemma}

At each vertex $(g,x)$, the $q_i+1$ direction-$i$ edges lead to
$(gb,x+e_i)$, $b\in B_i$. Choosing one such edge in each direction
of $I$ determines a unique $I$-face; see
Figure~\ref{fig:edges-and-cubes}.
The complex $\X$ is connected exactly when the steps $(b,e_i)$,
$b\in B_i$, $i\in[r]$, generate $G\times\F_2^r$.
We use lower and upper \emph{directional degree bounds}
$\underline q$ and $\overline q$ such that
$
 1\le\underline q\le q_i+1\le\overline q
$ for any $i\in[r]$.

The coface counts also give the total number of cells.
\begin{lemma}[Cell counts]\label{lem:cell-counts}
Each full vertex color occurs $|G|$ times, and
\begin{align*}
 F_k:=|\X(k)|=|G|2^{r-k}\nu_k,\qquad
 \nu_k:=\sum_{|I|=k}\prod_{i\in I}(q_i+1),
\end{align*}
where $\nu_k$ is also the number of $k$-faces containing a fixed vertex.
\end{lemma}
\begin{proof}
Apply Lemma~\ref{lem:cubical-incidence} to a vertex and sum
over direction sets of size $k$ to obtain $\nu_k$ faces through it.
Since every $k$-face has $2^k$ vertices, counting vertex-face
incidences gives
\[
 2^k F_k=|\X(0)|\nu_k=|G|2^r\nu_k.
\]
Each full color occurs $|G|$ times because
$\X(0)=G\times\F_2^r$.
\end{proof}

\section{Expansion framework for cubical sheaf complexes}
\label{sec:complex-expansion}\label{sec:local-to-global}
We now restrict to finite cubical complexes equipped with sheaf
codes. The definitions and one-sided expansion arguments apply to
general sheaf codes. Whenever $\cF^\perp$ is used, we assume that
$\cF$ is Tanner and use the dual of Example~\ref{ex:tanner-sheaf}.
For the expansion arguments, $\X$ is a decorated Cayley
cubical complex as in Section~\ref{sec:cubical-complex}. We follow
the framework of~\cite{DinurLinVidick2024}.
Recall from Lemma~\ref{lem:cell-counts} that $F_k=|\X(k)|$ and
$\nu_k=|\X_{\ge v}(k)|$ for every vertex $v$.
The directional degree bounds satisfy $1\le\underline q\le q_i+1\le\overline q$.

\subsection{Expansion parameters and standard reductions}
\label{sec:expansion-parameters}\label{sec:cominimal-expansion}
We follow closely the terminology in~\cite[Definition~2.1]{DinurLinVidick2024}.

\begin{definition}[Distances and expansion]
\label{def:systolic-distances}\label{def:cycle-expansion}
The degree-$k$ systolic and cosystolic distances are
\[
 \mu_{\rm syst}(k)=\min_{x\in\ker\partial_k\minus\im\partial_{k+1}}|x|_b,
 \qquad
 \mu_{\rm cosyst}(k)=\min_{x\in\ker\delta^k\minus\im\delta^{k-1}}|x|_b.
\]
The cycle and cocycle expansion constants are
\[
 \eps_{\rm cyc}(k)=\inf_{x\notin\ker\partial_k}
     \frac{|\partial_kx|_b}{\db(x,\ker\partial_k)},\qquad
 \eps_{\rm cocyc}(k)=\inf_{x\notin\ker\delta^k}
     \frac{|\delta^kx|_b}{\db(x,\ker\delta^k)}.
\]
\end{definition}

The support argument below is stated in terms of DLV's local co-minimality
\cite[Definition~2.2]{DinurLinVidick2024}.

\begin{definition}[Local co-minimality]\label{def:local-cominimality}
A cochain $x\in C^k(\X;\cF)$ is \emph{locally co-minimal} if
\[
 |x+\delta y|_b\ge|x|_b
\]
for every vertex $v$ and every $y\in C^{k-1}(\X;\cF)$ supported
on $\X_{\ge v}(k-1)$.
The locally co-minimal distance is
\[
 d_{\coloc}(k)=\min\{|x|_b:0\ne x\in\ker\delta^k,
                       \ x\text{ locally co-minimal}\},
\]
with value $+\infty$ if the set is empty.
\end{definition}

In degree zero the local co-minimality condition is vacuous.

\begin{lemma}[From local co-minimality to expansion,
{\cite[Lemma~2.3]{DinurLinVidick2024}}]
\label{lem:cominimal-to-expansion}
For $0\le k<r$, one has
$\mu_{\rm cosyst}(k)\ge d_{\coloc}(k)$.
For $0\le k\le r-2$,
\begin{equation}
 \eps_{\rm cocyc}(k)\ge
 \min\left\{\frac1{\nu_k},\,
              \frac{d_{\coloc}(k+1)}{F_k}\right\}.
 \label{eq:cominimal-to-expansion}
\end{equation}
\end{lemma}

Our Tanner sheaf and its dual are the generated sheaves of
\cite[Definitions~3.10--3.11]{LiEtAl2025}. We use their duality theorem.
\begin{proposition}[Sheaf duality,
{\cite[Theorem~3.17]{LiEtAl2025}}]
\label{prop:sheaf-duality}\label{sec:dual-sheaf}
Let $\cF$ be a Tanner sheaf on a finite pure $r$-dimensional cubical
complex $\X$, and let $\cF^\perp$ be its dual Tanner sheaf.
If both are locally acyclic, then, for $0\le k\le r$,
\[
 H^{r-k}(\X;\cF^\perp)\cong H_k(\X;\cF).
\]
The corresponding systolic/cosystolic distances and
cycle/cocycle expansion constants are of the same order up to factors
depending only on $r$ and $D$, where $D$ is the maximum number of
cofaces of dimension one higher containing a cell.
\end{proposition}

In this comparison, each expansion constant $\eps$ is replaced by
$\min\{1,\eps\}$. For our decorated cubical complexes, $D\le r\overline q$.

Finally, the following lemma gives lower bounds on the distance and
soundness of the associated CSS code in terms of the systolic and
cosystolic distances and the cycle and cocycle expansion constants,
respectively. It is the binary case of~\cite[Lemma~2.7]{DinurLinVidick2024},
with $H_X,H_Z$ interchanged.

\begin{lemma}[CSS parameters;
binary case of {\cite[Lemma~2.7]{DinurLinVidick2024}}]
\label{lem:css-parameters}
Let $C$ be a finite cochain complex over $\F_2$, with block
decompositions $C^j=\bigoplus_{\sigma\in\X(j)}V_\sigma$.
All dimensions in this lemma are over $\F_2$.
Choose bases on each $V_\sigma$ and dual bases on $V_\sigma^*$,
and let $B^j$ represent $\delta^j$.
If $\dim C^{k-1},\dim C^k,\dim C^{k+1}>0$, the checks $H_X=B^k$,
$H_Z=(B^{k-1})^{\mathsf T}$ give a binary CSS code with parameters
$[[N,K,D,\rho]]$, where
\[
 N=\dim C^k,\qquad K=\dim H^k(C),\qquad
 D\ge\min\{\mu_{\rm syst}(k),\mu_{\rm cosyst}(k)\},
\]
and $\rho$ is given by
\begin{equation}
 \rho=\frac{\dim C^k}{\max_{\sigma\in\X(k)}\dim V_\sigma}
 \min\left\{
 \frac{\eps_{\rm cocyc}(k)}{\dim C^{k+1}},\,
 \frac{\eps_{\rm cyc}(k)}{\dim C^{k-1}}\right\}.
 \label{eq:css-parameter-soundness}
\end{equation}
\end{lemma}

\subsection{Parallel-face walks and geometric expansion}
\label{sec:geometric-expansion}
We adapt mixing method in \cite[Section~6]{DinurLinVidick2024}
to parallel-face graphs, whose moves may change edge labels. Throughout this section, functions on faces are
real-valued, and all inner products and Markov operators
are taken over $\mathbb R$.

\begin{definition}[Parallel-face graphs and spectral expansion]
\label{def:parallel-face-graphs}\label{def:directional-expansion}
Fix $I\subsetneq[r]$, $j\notin I$, and the colors outside
$I\cup\{j\}$. The direction-$j$ parallel-face graph has the
$I$-cells of these colors as vertices. Each $(I\cup\{j\})$-cell
gives an edge between its two opposite $I$-faces, called parallel
neighbors. This graph is $(q_j+1)$-regular and bipartite by color $j$.
The complex $\X$ satisfies the \emph{spectral assumption} with
parameter $0\le\lambda<1$ if every such graph is connected and
its normalized eigenvalues other than $1$ and $-1$ have absolute
value at most $\lambda$.
\end{definition}

For $0\le\ell\le k<r$, define the walk $W$ on $\X(k)$ as in
\cite[Definition~6.6]{DinurLinVidick2024}. From $f\in \X(k)$, select an $\ell$-face
$v\le f$, then a direction unoccupied by $v$ and a parallel
neighbor $v'$ in that direction, and finally a $k$-face containing
$v'$. Each choice is uniform among the available options.
We also write $W$ for its transition operator.

\begin{lemma}[Mixing estimate,
adapted from {\cite[Lemma~6.12]{DinurLinVidick2024}}]
\label{lem:cubical-mixing}
Suppose $\X$ satisfies Definition~\ref{def:directional-expansion}
with bound $\lambda$.
Fix $0\le\ell\le k<r$ and $R\ge\overline q/\underline q$.
For the walk $W$ above and any $A\subseteq \X(k)$,
\begin{equation}
 \langle\one_A,W\one_A\rangle
 \le\lambda|A|+D_{k\ell}|A|^2/F_k,\qquad
 D_{k\ell}=2^{r-\ell-1}\binom r\ell R^k.
 \label{eq:walk-mixing}
\end{equation}
The inner product is the counting inner product on $\X(k)$.
\end{lemma}
\begin{proof}
We follow~\cite[Lemma~6.12]{DinurLinVidick2024} with the incidence
weights below. Write $W=UMD$, where $D$ averages over $k$-cofaces,
$M$ performs the parallel step, and $U$ averages over $\ell$-subfaces.
A $k$-face has
$b_{k\ell}=\binom k\ell2^{k-\ell}$ subfaces of dimension $\ell$.
An $\ell$-face of occupied type $I$ has
\[
 d_I=\sum_{J\supseteq I, |J|=k}\prod_{i\in J\minus I}(q_i+1)
\]
cofaces of dimension $k$. The probability measures
\[
 \pi_k(f)=1/F_k\quad(f\in \X(k)),\qquad
 \pi_\ell(v)=d_I/(b_{k\ell}F_k)\quad(v\in \X(I),\ |I|=\ell)
\]
make $D,U$ adjoint Markov contractions, by incidence counting.

The parallel step $M$ preserves type and is self-adjoint for
$\pi_\ell$. A fixed-direction parallel graph, with
its other colors fixed, has $\pi_\ell$-mass
\[
 \frac{2|G|(\prod_{i\in I}(q_i+1))d_I}{b_{k\ell}F_k}
 \ge\frac1{D_{k\ell}}.
\]
For the inequality use
$(\prod_{i\in I}(q_i+1))d_I/\nu_k
\ge[\binom k\ell/\binom r\ell]R^{-k}$.
Order the $r-\ell$ free directions separately for each
occupied type. For each position in these orders, take the operator
that uses the corresponding direction on each type. Averaging these
$r-\ell$ operators gives $M$.
On each connected parallel graph, the upper spectral bound gives
the nonconstant contribution at most $\lambda$ times squared norm.
For a nonnegative function the component mean terms sum to at most
$D_{k\ell}$ times its global mean squared: divide by each component
mass and use that a sum of squares of nonnegative integrals is at
most the square of their sum. The same bound therefore holds after averaging these operators.
For $g=D\one_A$, adjointness and contraction give
\[
 \mathbb E_{\pi_\ell}g=\frac{|A|}{F_k},
 \qquad
 \|g\|_{\pi_\ell}^2\le\frac{|A|}{F_k}.
\]
Consequently,
\[
 \frac{\langle\one_A,W\one_A\rangle}{F_k}
 =\langle g,Mg\rangle_{\pi_\ell}
 \le \lambda\frac{|A|}{F_k}
      +D_{k\ell}\left(\frac{|A|}{F_k}\right)^2,
\]
which proves \eqref{eq:walk-mixing}.
The
bipartite eigenvalue $-1$ causes no problem for this upper bound.

\end{proof}

\subsection{Local-to-global expansion}
\label{sec:local-robustness}\label{sec:sheaf-criterion}

We combine local coboundary expansion with the geometric mixing
estimate of Lemma~\ref{lem:cubical-mixing} to obtain global distance
and expansion bounds. We first state the local condition, in the
distance form of DLV's robustness condition
\cite[Definition~5.6 and Remark~5.7]{DinurLinVidick2024}.

\begin{definition}[Local coboundary expansion]
\label{def:local-coboundary-expansion}
A cubical sheaf $\cF$ has local coboundary expansion $\kappa>0$ if for
every face $\sigma$, every $0\le j<r-\dim\sigma$, and every
$y\in K_\sigma^j(\cF)$,
\begin{equation}
 \bigl|\delta_\sigma^j y\bigr|_b\ge
 \kappa\,\underline q\,\db\bigl(y,\im\delta_\sigma^{j-1}\bigr).
 \label{eq:local-coboundary-expansion}
\end{equation}
\end{definition}

Local expansion implies local acyclicity.
When all directional lengths agree, take $\underline q=\overline q$.
For tensor sheaves, expansion of both
$\cF$ and $\cF^\perp$ is DLV's \emph{two-way robustness}
\cite[Definition~5.9]{DinurLinVidick2024}.

Under this local condition, the next lemma gives a linear lower
bound on the locally co-minimal distance. Its proof adapts DLV's
support argument using unique cube completion and our parallel-face
mixing estimate.

\begin{lemma}[Locally co-minimal distance, adapted from {\cite[Proposition~7.1]{DinurLinVidick2024}}]\label{lem:global}
Fix $r,R,\kappa$. There are $\lambda_0,\eta>0$, depending only on
these parameters, such that if $\X$ has $\overline q/\underline q\le R$ and satisfies
Definition~\ref{def:directional-expansion} with $\lambda\le\lambda_0$, then local coboundary
expansion $\kappa$ implies, for every $k<r$,
\begin{equation}
 d_{\coloc}(k)\ge\eta F_k.
 \label{eq:coloc-bound}
\end{equation}
\end{lemma}
\begin{proof}
Following~\cite[Claims~7.4--7.6]{DinurLinVidick2024}, take a nonzero
locally co-minimal cocycle $x$, with block support
$A=\supp_b x\subseteq \X(k)$.
For an $\ell$-face $v$, its restriction to the upper star is locally
minimum: a local modification there extends by zero, and its
coboundary remains in that upper star, which lies in the upper star
of any vertex of $v$. In the $(k+1)$-cofaces of $v$, the local
syndrome is canceled by external $k$-faces since $\delta x=0$.
Split these external faces into $\operatorname{Nb}_v(k)$, meeting
$v$ in an $(\ell-1)$-face, and $\operatorname{Op}_v(k)$, disjoint
from $v$ and containing a parallel translate of it. Each external
face $f'$ occurs in at most one such coface: the colors of $v\cup f'$
span that coface's coordinate cube, so
Lemma~\ref{lem:cubical-incidence} determines it uniquely.
Thus~\eqref{eq:local-coboundary-expansion} yields
\begin{equation*}
 \kappa\underline q\,|A\cap \X_{\ge v}(k)|
 \le |A\cap\operatorname{Nb}_v(k)|+
 |A\cap\operatorname{Op}_v(k)|.
\end{equation*}
Coefficient maps may kill a nonzero block, but cannot increase this
support bound.

Let $c_r>0$ denote an incidence-counting constant depending only on
$r$. For a fixed $k$-face $f$, sum the preceding local inequality over its
$(\ell+1)$-subfaces. Each neighbor term can be assigned to its
intersection $\ell$-face $v\subseteq f$. For fixed $v$ and an external
$k$-face $f'$, the possible larger subface of $f$ is obtained by
adding a direction in $\operatorname{type}(f)\minus
\operatorname{type}(f')$. There is at most one choice in each such
direction. The multiplicity is at most $k-\ell$, and also at most
$r-k$ because $f'$ has $k$ occupied directions. This gives
a multiplicity bounded in terms of $r$. At $\ell=0$
there are no neighbor terms. Starting with $v=f$ and descending
therefore proves
\begin{equation*}
 \one_A(f)\le c_r
 \sum_{\ell=0}^k\frac{1}{(\kappa\underline q)^{k+1-\ell}}
 \sum_{\substack{v\subseteq f\\\dim v=\ell}}
 |A\cap\operatorname{Op}_v(k)|.
\end{equation*}

Each pair $(v,f')$ in the last sum admits a path of the walk $W$:
choose $v$, move it parallel in one free direction, then choose
$f'$ as a coface. A specified such path has probability at least
\[
 \frac1{b_{k\ell}(r-\ell)\binom{r-\ell}{k-\ell}
 \overline q^{k+1-\ell}}.
\]
Here $b_{k\ell}=\binom k\ell2^{k-\ell}$ counts the $\ell$-faces of
a $k$-face, as in the proof of Lemma~\ref{lem:cubical-mixing}.
Sum this recursion over $f\in A$ and apply
\eqref{eq:walk-mixing}. The factors $\overline q^{k+1-\ell}$ cancel
against $\underline q^{k+1-\ell}$ up to powers of $R$.
All remaining incidence multiplicities depend only on $r$.
Thus, for some $C=C(r,R,\kappa)\ge1$,
\begin{equation*}
 |A|\le C\lambda|A|+C|A|^2/F_k.
\end{equation*}
Choose $\lambda_0$ so that $C\lambda_0<1$ and let
$\eta=(1-C\lambda_0)/C$.
For $A\ne\varnothing$ this proves~\eqref{eq:coloc-bound}.
\end{proof}

\begin{corollary}[Local-to-global criterion]\label{thm:sheaf-criterion}
Fix $r\ge4$, $R\ge1$ and $\kappa>0$. There is
$0<\lambda_0=\lambda_0(r,R,\kappa)<1$ with the following property.
Consider a family of connected decorated Cayley cubical complexes
$\X$, with sheaves $\cF$ generated by codimension-one local codes over a finite
characteristic-two field. Assume the field and directional degrees
are fixed as $|\X(0)|\to\infty$, and:
\begin{enumerate}[label=(\roman*),leftmargin=*]
\item $\overline q/\underline q\le R$, and $\X$ satisfies the spectral assumption of
Definition~\ref{def:directional-expansion} with parameter $\lambda_0$.
\item Both $\cF$ and $\cF^\perp$ have local coboundary expansion $\kappa$.
\end{enumerate}
Then, for every $2\le k\le r-2$ with $\dim_F C^k>0$,
\begin{equation}
 \mu_{\rm syst}(k),\mu_{\rm cosyst}(k)=\Omega(\dim_F C^k),\qquad
 \eps_{\rm cyc}(k),\eps_{\rm cocyc}(k)=\Omega(1).
 \label{eq:global-expansion}
\end{equation}
\end{corollary}
\begin{proof}
For either sheaf, Lemmas~\ref{lem:global}
and~\ref{lem:cominimal-to-expansion} give
$\mu_{\rm cosyst}(k)\ge\eta F_k$ and, for $k\le r-2$ and every $x\in C^k$,
\begin{equation}
 |\delta x|_b\ge
 \min\{1/\nu_k,\eta F_{k+1}/F_k\}\,
 \db(x,\ker\delta^k).
 \label{eq:all-word-global}
\end{equation}
The coefficient is constant because the directional degrees are fixed.
Local acyclicity and Proposition~\ref{prop:sheaf-duality} transfer the
dual-sheaf bounds in degree $r-k$ to the chain parameters.
Finally, cell counting and $\dim\cF(\sigma)=O(1)$ give
$F_{r-k}=\Theta(F_k)$ and $\dim C^k=O(F_k)$.
\end{proof}

Section~\ref{sec:realization} combines this criterion with positive
rate and Lemma~\ref{lem:css-parameters}.

\section{Sheaf codes on products of trees and their quotients}
\label{sec:rs-construction}
\label{sec:coefficients}\label{sec:descent-certificate}

In this section we construct the finite sheaf codes, in three steps.
The first is abstract: from arbitrary local codes on infinite regular trees
we build a tensor-product sheaf on a product of trees, and give conditions
under which such a sheaf descends to a finite quotient. To
instantiate the framework, we take the quotients from arithmetic lattices acting on products of
Bruhat-Tits trees, following Rungtanapirom, Stix, and
Vdovina~\cite[Sections~2 and~6]{RSV2019}:
$$
    \X=\Gamma\backslash\mathcal T,
    \qquad
    \mathcal T=\mathcal T_1\times\cdots\times\mathcal T_r,
$$
where $\mathcal T_i$ is the $(q_i+1)$-regular Bruhat-Tits tree of the $i$th
place and $\Gamma$ a quaternionic lattice acting simply transitively on the
vertices of $\mathcal T$.
Finally, we take the local codes to be (twisted) Reed-Solomon codes and check that
they are compatible with the arithmetic action.

\subsection{Tree codes, tensor products, and descent}
\label{sec:tree-code-construction}
\paragraph{Sheaves on infinite trees.}
Fix a finite field $F$ of characteristic two, and let $\mathcal T_i$
be an infinite $(q_i+1)$-regular tree. We view each tree as a
one-dimensional cell complex: vertices are $0$-cells, edges are
$1$-cells, and $v\le e$ means that $v$ is an endpoint of $e$. 
Write $E_i(v)$ for the set of edges incident to
$v\in\mathcal T_i(0)$. In this case, a sheaf code $\cF_i$ on $\mathcal T_i$ is always a Tanner sheaf code, and is hence specified
by arbitrary linear codes $C_i(v)\subseteq F^{E_i(v)}$, with
\[
 \cF_i(v)=C_i(v),\qquad \cF_i(e)=F,\qquad
 \rho^i_{v,e}(c)=c_e.
\]
Thus a vertex carries a codeword and each incident edge reads one
coordinate.
At a vertex $v$, the local complex is
$
 K_v(\cF_i)=[C_i(v)\hookrightarrow F^{E_i(v)}].
$
The two terms lie in degrees zero and one. At an edge $e$, the local complex $K_e(\cF_i)$ consists of
$F$ in degree zero.

\paragraph{Tensor products.}
The product $\mathcal T=\prod_{i=1}^r\mathcal T_i$ has a cubical complex structure: a cell $\sigma=\sigma_1\times\cdots\times\sigma_r$
has a vertex or an edge in each factor, and face inclusion is
coordinatewise. Its type is $I=\{i:\sigma_i\text{ is an edge}\}$
and its dimension is $|I|$. Writing $v_i=\sigma_i$ for $i\notin I$,
we identify its top-dimensional cofaces as
$U_\sigma\cong\prod_{i\notin I}E_i(v_i)$.

The product sheaf code $\mathcal G=\boxtimes_i\cF_i$ is defined by
\cite[Section~V]{PanteleevKalachev2024}
\[
 \begin{aligned}
 \mathcal G(\sigma):=\bigotimes_{i=1}^r\cF_i(\sigma_i)
       \cong\bigotimes_{i\notin I}\mathcal{F}_i(v_i)\subseteq F^{U_\sigma},\qquad
 \rho_{\sigma,\tau}:=\bigotimes_{i=1}^r\rho^i_{\sigma_i,\tau_i}:
 \mathcal G(\sigma)\longrightarrow\mathcal G(\tau)\qquad
(\sigma\le\tau).
 \end{aligned}
\]
Its local complex is
\begin{equation}
 K_\sigma(\mathcal G)\cong\bigotimes_{i\notin I}
 [C_i(v_i)\hookrightarrow F^{E_i(v_i)}].
 \label{eq:tree-code-local-model}
\end{equation}
The tensor product denotes the total cochain complex of
these vertex local complexes. Since each factor is exact in degree zero, the total complex is also exact below its top degree, and hence $\mathcal G$ is locally acyclic.

\paragraph{Passing to finite quotients.}
Let $\Gamma$ act on $\mathcal T$ by direction-preserving cubical
automorphisms. That is each $g\in\Gamma$ acts coordinatewise on $\mathcal T(0) = \prod_{i=1}^r \mathcal T_i(0)$ and extends to the whole cubical complex structure. In particular, a direction-$i$ edge remains a direction-$i$ edge. 
Let $\mathcal G$ be a sheaf code on $\mathcal T$.

Passing to the quotient identifies cells related by $\Gamma$.
Their coefficient spaces and restriction maps must also
be identified consistently. In particular, restricting
a value and then identifying must give the same result
as identifying first and then restricting.
The following definition expresses this requirement.

\begin{definition}[Group actions on sheaf codes]
\label{def:sheaf-action}
An \emph{action} of $\Gamma$ on $\mathcal G$ is a lift of the
given action on $\mathcal T$ to the local code spaces by
$F$-linear isomorphisms
$g:\mathcal G(\sigma)\to\mathcal G(g\sigma)$,
$c\mapsto g\cdot c$, commuting with restrictions:
\[
 \rho_{g\sigma,g\tau}(g\cdot c)
 =g\cdot\rho_{\sigma,\tau}(c)
 \qquad(g\in\Gamma,\ \sigma\le\tau,\ c\in\mathcal G(\sigma)).
\]
\end{definition}

For the product $\mathcal G=\boxtimes_i\cF_i$, let $\Gamma_i$ be
the image of $\Gamma$ in $\operatorname{Aut}(\mathcal T_i)$.
It suffices to lift the $\Gamma_i$-action on $\mathcal T_i$ to
$\cF_i$ in each direction. An element $g\in\Gamma$ then acts on
$\mathcal G$ by tensoring the sheaf actions of its components
$g_i\in\Gamma_i$. It is clear that compatibility with restrictions holds since each factor does.

\begin{proposition}[Descent of sheaf codes]
\label{prop:tree-code-descent}
Let $\Gamma$ act on a sheaf code $\mathcal G$ on $\mathcal T$, with its action
on $\mathcal T$ free on vertices and having finitely many vertex
orbits. If the quotient map $p:\mathcal T\to\X:=\Gamma\backslash\mathcal T$
is injective on each closed cube, then $\X$ is a finite $r$-dimensional cubical
complex and $\mathcal G$ descends to a sheaf code $\cF$ on $\X$.
For every cell $\sigma$ of $\mathcal T$, there is an isomorphism
of cochain complexes preserving block weight:
\[
 K_{p(\sigma)}(\cF)\cong K_\sigma(\mathcal G).
\]
In particular, if $\mathcal G$ is a Tanner sheaf, then so is $\cF$.
\end{proposition}
\begin{proof}
Cube injectivity gives embedded closed cubes in the quotient.
It also implies that a cell stabilizer fixes every vertex of that cell and is therefore
trivial. The finite number of vertex orbits
make $\X$ finite. We next identify the upper stars of $\sigma$ and
$p(\sigma)$. Every quotient coface lifts to one containing $\sigma$.
If $\tau,\tau'\ge\sigma$ and $p(\tau)\le p(\tau')$, choose $g$
with $g\tau\le\tau'$. For a vertex $v$ of $\sigma$, both $v$ and
$gv$ lie in $\tau'$ and have the same image under $p$. Cube
injectivity gives $gv=v$, and vertex freeness gives $g=1$.
Thus the coface lifts are unique and preserve face inclusions.

For a quotient cell $\bar\sigma$, identify the spaces
$\mathcal G(\sigma)$ on all its lifts using the group action, and
denote the resulting vector space by $\cF(\bar\sigma)$.
Trivial stabilizers make these identifications unambiguous, and
equivariance makes the induced restriction maps independent of the lifts.
Choose a representative of each top-cell orbit to identify its
descended space with $F$. Restrictions to these top-cell spaces
embed $\cF(\bar\sigma)$ into $F^{U_{\bar\sigma}}$: after choosing
a lift $\sigma$, they are the original coordinate restrictions of
$\mathcal G(\sigma)$ followed by isomorphisms on the one-dimensional
targets. Thus $\cF$ is a sheaf code in the sense of
Definition~\ref{def:sheaf-code}.

The upper-star bijection and these identifications give the claimed
cochain isomorphism. It preserves block weight because each coface
space maps isomorphically to its descended space.
These identifications also preserve the codimension-one constraints,
so the descent of a Tanner sheaf is Tanner.
\end{proof}
\paragraph{Tanner sheaves and duals.}
The tensor product sheaf $\mathcal G=\boxtimes_i\cF_i$ is itself a
Tanner sheaf. A codimension-one cell $\tau$ has a vertex factor $v$ in
a single direction $j$, with $U_\tau\cong E_j(v)$. We then assign the code $\mathcal C_\tau=C_j(v)$ to the cell. For a cell $\sigma$ of type $I$, the sets
$U_\tau\subseteq U_\sigma\cong\prod_{i\notin I}E_i(v_i)$ over its
codimension-one cofaces $\tau$ are exactly the coordinate lines, so the
generating constraints at $\sigma$ cut out the tensor product code
$\bigotimes_{i\notin I}C_i(v_i)=\mathcal G(\sigma)$.

Replacing each $C_i(v)$ by $C_i(v)^\perp$ gives the Tanner dual
$\mathcal G^\perp$, which descends to $\cF^\perp$ by
Proposition~\ref{prop:tree-code-descent}.

\subsection{Products of trees and cubical quotients}
\label{sec:finite-presentation}
We construct finite quotients of products of Bruhat-Tits trees.
A common quaternion algebra supplies a group acting in all directions.

\paragraph{A single tree and its neighbors.}
The lattice model identifies the edges at each vertex with projective
points, which will serve as evaluation points for the local codes.
Put $k_0=\F_2(y)$, and let $P\in\F_2[y]$ be a monic irreducible
of positive even degree $a$. Write the discrete valuation $\nu=v_P$, uniformizer $\pi=P$, and $q=2^a$. Then, the valuation ring $\mathcal O_\nu$ of $(k_0)_\nu$ has residue field $\F_q$.

We will model the tree via the Bruhat-Tits tree $\mathcal T_\nu$, which has as vertices the equivalence classes of rank-two
$\mathcal O_\nu$-lattices in $(k_0)_\nu^2$, up to multiplication
by a nonzero scalar. For a fixed representative $\mathcal L$ of
$v$, its neighbors have unique representative lattices $\mathcal M$ with
\[
 \pi\mathcal L\subsetneq\mathcal M\subsetneq\mathcal L.
\]
Reduction $\mathcal M\mapsto\mathcal M/\pi\mathcal L$ identifies
these neighbors with the lines in $\mathcal L/\pi\mathcal L$.
Choosing coordinates $\mathcal L/\pi\mathcal L\simeq\F_q^2$
therefore identifies $E(v)\cong\PP^1(\F_q)$, so the tree is
$(q+1)$-regular\footnote{For a proof that this graph is a tree, see
\cite[Chapter~II, \S1.1]{Serre1980}.}.
The group $\PGL_2((k_0)_\nu)$ acts by $g[\mathcal L]=[g\mathcal L]$.
We write $o_\nu=[\mathcal O_\nu^2]$ for the base vertex,
the standard lattice class.

\paragraph{Vertex colors.}
We compute determinants in the fixed ambient coordinates of $(k_0)_\nu^2$.
For a lattice $\mathcal L$ with an $\mathcal O_\nu$-basis $(b_1,b_2)$, the coloring is given by
$$
 \operatorname{color}([\mathcal L])
 \coloneqq \nu\bigl(\det(\mathcal L)\bigr)\bmod2,\qquad
 \det(\mathcal L)\coloneqq\det(b_1,b_2)\mathcal O_\nu.
$$
Changing basis multiplies the determinant by a unit, leaving its valuation
unchanged. Scaling a lattice changes this valuation by an even integer,
so the color is well defined. Adjacent vertices have opposite colors.
Note that $o_\nu=[\mathcal O_\nu^2]$ is of color $0$ and a linear automorphism $g
\in\mathrm{GL}_2((k_0)_\nu)$ changes the color by
$\nu(\det g)\bmod2$.

\paragraph{The product of trees.}
Let $P_1,\ldots,P_r\in\F_2[y]$ be distinct monic irreducibles
of positive even degrees $a_i$.
At each place $\nu_i=v_{P_i}$, take the corresponding tree
$\mathcal T_{\nu_i}$, with residue field
$\F_{q_i}$ of size $q_i=2^{a_i}$.
Put $\Sigma=\{\nu_1,\ldots,\nu_r\}$.
The product $\mathcal T_\Sigma=\prod_{\nu\in\Sigma}\mathcal T_\nu$
has one cubical direction for each chosen place. The base vertex is
$\mathbf o=([\mathcal O_\nu^2])_{\nu\in\Sigma}$.

\paragraph{The arithmetic action.}
Write $\F_4=\F_2(\theta)$ with $\theta^2+\theta+1=0$.
Set $k_1=k_0(\theta)$ and let $\varphi$ fix $k_0$ and send
$\theta$ to $\theta^2$. The quaternion algebra of
\cite[Section~2]{RSV2019} is $A(k_0)=k_1\oplus k_1z$, where
$z^2=1+y$ and $zc=\varphi(c)z$ for $c\in k_1$.
Its reduced norm is
$\Nrd(f+gz)=f\varphi(f)-(1+y)g\varphi(g)$.
Since each $a_i$ is even, choose an embedding
$\F_4\subseteq(k_0)_{\nu_i}$. The local splitting
$
 A((k_0)_{\nu_i}):=A(k_0)\otimes_{k_0}(k_0)_{\nu_i}
 \simeq M_2((k_0)_{\nu_i})
$
is given by
\[
 \iota_{\nu_i}(f+gz)=
 \begin{pmatrix}f&(1+y)g\\ \varphi(g)&\varphi(f)\end{pmatrix}.
\]
Its determinant is the reduced norm. Since $1+y$ is a unit at
each $\nu_i$, it also identifies the integral coefficient algebra
with $M_2(\mathcal O_{\nu_i})$.
Using the matrix action above in each factor,
$A(k_0)^\times/k_0^\times$ acts by
\[
 [a]\cdot([\mathcal L_i])_{i=1}^r
 =([\iota_{\nu_i}(a)\mathcal L_i])_{i=1}^r.
\]
The same global quaternion $a$ supplies the local matrix in every
direction. This action preserves directions. We next choose the
arithmetic subgroup used to form the quotients.

\paragraph{The generators.}
Define
\[
 S_{\nu_i}=\{[f+gz]:f,g\in\F_4[y],\ f\text{ monic},\quad
                         \Nrd(f+gz)=P_i\},\qquad
 \Gamma_\Sigma=\langle S_\nu:\nu\in\Sigma\rangle,
\]
where brackets denote classes modulo $k_0^\times$.
A generator in $S_{\nu_i}$ has integral matrices at all selected
places, with determinant $P_i$. It therefore moves $o_{\nu_i}$
to the neighbor given by its image line modulo $P_i$, and fixes
the base vertices in the other factors.
Lemma~\ref{lem:cubical-generators} shows that every line occurs
once, $S_{\nu_i}$ is inverse-closed with $|S_{\nu_i}|=q_i+1$,
and $\Gamma_\Sigma$ acts simply transitively on vertices.

\paragraph{The color-preserving subgroup.}
Use the vertex coloring defined above in each
factor. Since the determinant of each local action is the reduced norm,
the map
$\operatorname{color}([a])=(\operatorname{ord}_\nu\Nrd(a)\bmod2)_{\nu\in\Sigma}
\in\F_2^\Sigma$ records the color changes.
Set $\Gamma_\Sigma^0=\ker\operatorname{color}$.
The norm of a generator in $S_{\nu_i}$ is $P_i$, so
\begin{equation}
 \Gamma_\Sigma^0
 =\{[a]\in\Gamma_\Sigma:\Nrd(a)\in(k_0^\times)^2\}.
 \label{eq:type-norm-one}
\end{equation}
Dividing by a square root gives a norm-one representative.
In characteristic two this representative
is unique, so the local matrix realizations give a homomorphism
$\Gamma_\Sigma^0\to\prod_i\SL_2((k_0)_{\nu_i})$.

\paragraph{Congruence quotients.}
Choose a monic irreducible $\ell_s\in\F_2[y]$ coprime to
$(1+y)\prod_iP_i$, with degree $s>2r\max_i a_i$.
Put $I_s=(\ell_s)$ and $L_s=\F_2[y]/I_s\simeq\F_{2^s}$.
Reducing quaternion coefficients modulo $I_s$ and splitting the
resulting algebra as $M_2(L_s)$ gives
$\pi_{I_s}:\Gamma_\Sigma\to\PGL_2(L_s)$.
An explicit matrix realization is given in
Appendix~\ref{sec:quotient-injectivity}.
Define the color-preserving congruence subgroup and its quotient by
\[
 \Gamma_\Sigma(I_s)=\{g\in\Gamma_\Sigma^0:\pi_{I_s}(g)=1\},
 \qquad \X_s=\Gamma_\Sigma(I_s)\backslash\mathcal T_\Sigma.
\]
Put $G_s=\PSL_2(L_s)=\PGL_2(L_s)$ and
$B_{\nu_i,I_s}=\pi_{I_s}(S_{\nu_i})$.
Proposition~\ref{prop:finite-quotients} identifies the vertex set
with $G_s\times\F_2^r$, with direction-$i$ steps
$(h,\varepsilon)\mapsto(hb,\varepsilon+e_i)$ for
$b\in B_{\nu_i,I_s}$, as in Section~\ref{sec:cubical-complex}.

\begin{proposition}[Arithmetic cubical family]
\label{prop:cubical-input}\label{prop:face-spectrum}
The complexes $\X_s$ are connected, have directional degrees
$q_i+1$, and satisfy
\begin{equation}
 \X_s\simeq\operatorname{Cay}_\square
 (G_s;B_{\nu_1,I_s},\ldots,B_{\nu_r,I_s}).
 \label{eq:finite-cell-definition}
\end{equation}
The quotient map $p_s:\mathcal T_\Sigma\to\X_s$
is injective on each closed cube.
The spectral assumption of Definition~\ref{def:directional-expansion}
holds with
$\lambda\le\max_i 2\sqrt{q_i}/(q_i+1)$.

For fixed $P_1,\ldots,P_r$,
\begin{equation}
 |\X_s(0)|=2^{r+s}(2^{2s}-1)=2^{\Theta(s)},\qquad
 |\X_s(k)|=\Theta(|\X_s(0)|)\quad(0\le k\le r).
 \label{eq:quotient-size}
\end{equation}
The groups, cells and incidences are computable in time polynomial
in $|\X_s(0)|$.
\end{proposition}
\begin{proof}
Appendix~\ref{sec:quotient-injectivity} proves the group and
covering claims using Lemma~\ref{lem:arithmetic-approximation}.
Appendix~\ref{sec:spectral-certificate} identifies each parallel-face
graph with a congruence tree quotient and applies
Lemma~\ref{lem:arithmetic-ramanujan}, with vertex or edge stabilizers
in the transverse directions.
\end{proof}

The parameters $q_i$ and $s$ have different roles.
The field sizes $q_i$ determine the local degrees and
the spectral bound, whereas $s$ controls the size of
the finite quotient. Once the local trees are fixed,
increasing $s$ produces larger complexes while preserving
their upper stars. This will allow the code length to grow
with the local codes held fixed.

\subsection{Reed-Solomon sheaves}
\label{sec:directional-codes}\label{sec:edge-matching}
\label{sec:rs-compatibility}\label{sec:rs-descent}
All spaces and codes below are over a finite field $F$ containing
$\F_{q_i^2}$ for every direction $i$.
Write $F[X,Y]_d$ for the space of homogeneous polynomials
of degree $d$ over $F$.

For a subfield $\F_q\subseteq F$ and $0\le d\le q-1$, define the
degree-$d$ \emph{projective Reed-Solomon code} by
\[
 \operatorname{PRS}_q(d)=
 \left\{\bigl((f(1,t))_{t\in\F_q},\,f(0,1)\bigr):
              f\in F[X,Y]_d\right\}\subseteq F^{q+1}.
\]
The evaluation vectors represent all points of $\PP^1(\F_q)$.
Changing their order or their nonzero representatives gives a
\emph{monomially equivalent} code: one obtained by permuting coordinates
and multiplying each coordinate by a nonzero scalar.
The edges incident to each vertex of $\mathcal T_i$ are indexed by
$\PP^1(\F_{q_i})$ (Section~\ref{sec:finite-presentation}), so we use
these codes as the directional codes.

\paragraph{Generating codes on a single tree.}
Fix a place $\nu\in\Sigma$, with uniformizer $\pi$ and residue-field
size $q$ as in Section~\ref{sec:finite-presentation}. Put $Q=\sqrt q$
and $G=\SL_2((k_0)_\nu)$.
Choose an integer $d$ satisfying
\begin{equation}
 0\le d\le q-1,\qquad Q-1\mid d.
 \label{eq:admissible-degrees}
\end{equation}
Fix the standard edge $e=\{o_\nu,o_\nu'\}$ on $\mathcal T_\nu$. Its endpoints
\[
 o_\nu=[\mathcal O_\nu^2] = [\mathcal O_\nu e_1\oplus\mathcal O_\nu e_2],\qquad
 o_\nu'=[\mathcal O_\nu e_1\oplus\pi\mathcal O_\nu e_2],
\]
have colors $0$ and $1$, respectively. Start with
\begin{equation}
 \cF(o_\nu)\xrightarrow{\ r_{o_\nu}\ }\cF(e)
 \xleftarrow{\ r_{o_\nu'}\ }\cF(o_\nu'),
 \label{eq:standard-edge-data}
\end{equation}
where $\cF(o_\nu)=\cF(o_\nu')=F[X,Y]_d$, $\cF(e)=F$,
$r_{o_\nu}(f)=f(1,0)$, and $r_{o_\nu'}(f)=f(0,1)$.
The two evaluation points are the residue directions of $e$ seen from
$o_\nu$ and $o_\nu'$, read in the respective lattice bases
$(e_1,e_2)$ and $(e_1,\pi e_2)$.

\begin{lemma}[Group action on a single tree sheaf]
\label{lem:matching-characters}\label{prop:tree-coefficients}
The standard-edge data~\eqref{eq:standard-edge-data} extend to a
$G$-equivariant Tanner sheaf $\cF$ on $\mathcal T_\nu$.
Moreover, at each vertex $v$, the restrictions identify $\cF(v)$ with
a code monomially equivalent to $\operatorname{PRS}_q(d)$.
\end{lemma}
\begin{proof}
For each $g\in G$, place a separate \emph{copy} of
$\cF(o_\nu)$, $\cF(e)=F$, and $\cF(o_\nu')$
at $go_\nu$, $ge$, and $go_\nu'$, respectively.
Write $f^{(g)}$ for the polynomial with the same coefficients in a
copied vertex space, and $1^{(g)}$ for the copied edge basis.
The transport maps are $f\mapsto f^{(g)}$ and $t\mapsto t\,1^{(g)}$.
Define the copied restrictions by
\[
 \begin{aligned}
 r_{go_\nu,ge}(f^{(g)})&=f(1,0)\,1^{(g)},\qquad
 r_{go_\nu',ge}(f^{(g)})=f(0,1)\,1^{(g)}.
 \end{aligned}
\]
By edge transitivity and color preservation
(Appendix~\ref{sec:tree-edge-transitivity}), these copies cover the tree.

To identify copies at the same cell, let
$K_{o_\nu}=\operatorname{Stab}_G(o_\nu)$ and
$K_{o_\nu'}=\operatorname{Stab}_G(o_\nu')$ act by
\[
 \begin{aligned}
 (k\cdot f)(X,Y)&=f\bigl((\bar k^{(Q)})^{-1}(X,Y)\bigr),
 &k\in K_{o_\nu},\quad f\in\cF(o_\nu),\\
 (k\cdot f)(X,Y)&=f(\bar k^{-1}(X,Y)),
 &k\in K_{o_\nu'},\quad f\in\cF(o_\nu').
 \end{aligned}
\]
Here $\bar k$ is the action of $k$ on lattice vectors modulo $\pi$,
written in the corresponding basis over $\F_q$ ($(e_1,e_2)$ at $o_\nu$ and $(e_1,\pi e_2)$ at $o_\nu'$). The superscript $(Q)$
raises the coefficients of this linear substitution to their $Q$th powers.
We then identify the vertex copies: fix one representative $g_v$ for each
vertex $v$, with $g_vv_0=v$ for the base vertex $v_0\in\{o_\nu,o_\nu'\}$ of
the same color, and set $f^{(g_vk)}=(k\cdot f)^{(g_v)}$ for $k\in K_{v_0}$.

For $h\in K_e:=K_{o_\nu}\cap K_{o_\nu'}$, the determinant-one condition
gives reciprocal scalar actions $a,a^{-1}$ on the two residue lines of $e$: $\langle\bar e_1\rangle$ seen from
$o_\nu$ and $\langle\overline{\pi e_2}\rangle$ seen from $o_\nu'$. Hence
\[
 r_{o_\nu}(h\cdot f)=f(a^{-Q},0)=a^{-Qd}r_{o_\nu}(f),\qquad
 r_{o_\nu'}(h\cdot f)=f(0,a)=a^d r_{o_\nu'}(f).
\]
Since $a^{q-1}=1$ and $q-1=(Q-1)(Q+1)$, the hypothesis $(Q-1)\mid d$
makes the two multipliers agree.

Set $h\cdot t=a^dt$ on $\cF(e)$ and identify the edge copies by
$1^{(gh)}=a^d1^{(g)}$. The stabilizer group laws make the identifications
consistent, leaving one copy at each cell. The $K_e$-equivariance makes the
restrictions well-defined. Relabelling copies by
$b\cdot f^{(g)}=f^{(bg)}$ and $b\cdot1^{(g)}=1^{(bg)}$ preserves these
identifications and defines a $G$-action commuting with the restrictions.
Identifying the standard edge data with the copies labelled $1$ gives
$g\cdot f=f^{(g)}$ and $g\cdot1=1^{(g)}$.
The edge basis $1^{(g)}$ depends on the representative of $ge$;
fix one translating element for each edge and use it in every degree (used for the later duality argument in Proposition~\ref{prop:rs-sheaf}).

The restrictions at a vertex $v$ then give
\[
 \cF(v)\longrightarrow F^{E(v)},\qquad
 f\longmapsto(r_{v,e'}(f))_{e'\in E(v)},
\]
where $E(v)$ is the set of incident edges. At $o_\nu$, write each
incident edge as $ke$ with $k\in K_{o_\nu}$. In the edge basis $k\cdot1$,
equivariance gives
\[
 r_{o_\nu,ke}(f)=r_{o_\nu}(k^{-1}\cdot f)
                 =f\bigl(\bar k^{(Q)}(1,0)\bigr).
\]
At $o_\nu'$ the same calculation gives $f\bigl(\bar k(0,1)\bigr)$.
As $ke$ ranges over the incident edges, these vectors represent every
projective direction once.
By transitivity, the same holds at every vertex.
Thus the image $C(v)\subseteq F^{E(v)}$ is monomially equivalent
to $\operatorname{PRS}_q(d)$.
Since $d<q+1$, the evaluation map is injective,
so it identifies $\cF(v)$ with $C(v)$ and the restrictions with coordinate
projections. Hence $\cF$ is a (Tanner) sheaf.
\end{proof}

\paragraph{Taking products and descending.}
Apply the single-tree construction in each direction $i$, with
$\nu=\nu_i$, $q=q_i$, $Q=Q_i:=\sqrt{q_i}$, and $d=d_i$
satisfying~\eqref{eq:admissible-degrees}.
Write $\cF_i$ for the resulting tree sheaf, and set
\[
 \mathcal G=\boxtimes_{i=1}^r\cF_i.
\]
Tensoring the actions from Lemma~\ref{lem:matching-characters} gives an
action of $\prod_i\SL_2((k_0)_{\nu_i})$ on $\mathcal G$. Restricting to
$\Gamma_\Sigma(I_s)\subseteq\Gamma_\Sigma^0$
gives the equivariant structure used for descent to $\X_s$.

Projective evaluations match the labels of the tree edges.
For the product-expansion argument, it is useful to express
the same local codes as ordinary Reed-Solomon codes on
norm-one evaluation sets.
The next lemma gives this description for both the codes
and their duals.

For a finite set $S\subseteq F$ and $0\le m\le|S|$, define the
ordinary \emph{Reed-Solomon code}
\[
 \RS(S,m)=\{(f(a))_{a\in S}:f\in F[X],\ \deg f<m\}.
\]
Set $\Sigma_i=\{z\in\F_{q_i^2}:z^{q_i+1}=1\}$.

\begin{lemma}[Projective and ordinary Reed-Solomon codes]
\label{lem:directional-evaluations}
For $0\le d\le q_i-1$, the code $\operatorname{PRS}_{q_i}(d)$ is
monomially equivalent to $\RS(\Sigma_i,d+1)$.
Its dual is $\operatorname{PRS}_{q_i}(q_i-1-d)$, hence monomially
equivalent to $\RS(\Sigma_i,q_i-d)$.
\end{lemma}
\begin{proof}
Choose $\xi_i\in\F_{q_i^2}\setminus\F_{q_i}$. The projective linear map
$[a:b]\mapsto[a+\xi_i b:a+\xi_i^{q_i}b]$ sends
$\PP^1(\F_{q_i})$ bijectively onto $\{[z:1]:z\in\Sigma_i\}$:
the ratio has $q_i$th power equal to its inverse, and both sets have
$q_i+1$ points. An invertible change of the two homogeneous variables
preserves $F[X,Y]_d$, while normalizing the evaluation vectors to
$(z,1)$ only rescales coordinates. The polynomials $f(X,1)$ then run
through all univariate polynomials of degree at most $d$, giving
$\RS(\Sigma_i,d+1)$.

Projective RS duality \cite[Corollary~4.11]{GimenezRuanoSanJose2023}
holds over $F$ by extension of scalars. Rescaling representatives preserves
this duality, since $c^{q_i-1-d}=c^{-d}$ for $c\in\F_{q_i}^\times$.
Applying the same coordinate change in degree $q_i-1-d$ gives
the RS description of the dual.
\end{proof}

\begin{proposition}[Reed-Solomon Tanner sheaves]
\label{prop:rs-sheaf}
For degrees satisfying~\eqref{eq:admissible-degrees}, the product
Tanner sheaf $\mathcal G$ descends to a Tanner sheaf $\cF$ on $\X_s$.
Its generating codes in direction $i$ are monomially equivalent
to $\RS(\Sigma_i,d_i+1)$.
Every cell $\sigma$ of type $I$ satisfies
$\dim\cF(\sigma)=\prod_{j\notin I}(d_j+1)$.
Replacing $d_i$ by $q_i-1-d_i$ gives the Tanner dual $\cF^\perp$.
Both $\cF$ and $\cF^\perp$ are locally acyclic.
\end{proposition}
\begin{proof}
The product of the actions in Lemma~\ref{lem:matching-characters}
restricts to the deck group $\Gamma_\Sigma(I_s)$
by~\eqref{eq:type-norm-one}.
Simple transitivity of $\Gamma_\Sigma$ and
Proposition~\ref{prop:cubical-input} give the geometric hypotheses
of Proposition~\ref{prop:tree-code-descent}, hence descent.
The claimed RS generating codes follow from
Lemmas~\ref{lem:matching-characters} and~\ref{lem:directional-evaluations}.
The tensor local chain complex
\eqref{eq:tree-code-local-model}, preserved under descent, gives
the dimension formula.
Replacing $d_i$ by $q_i-1-d_i$ gives $C_i(v)^\perp$ by
Lemma~\ref{lem:directional-evaluations}. An element $h\in K_e$ acts on
$t\in\cF(e)$ by $t\mapsto a^{d_i}t$ in degree $d_i$ and by
$t\mapsto a^{q_i-1-d_i}t=a^{-d_i}t$ in the complementary degree.
The two scalar factors cancel on the transported edge bases, so
the complementary-degree construction descends to $\cF^\perp$ (different representative $gh$, $h\in K_e$, rescales the two edge bases by
$a^{d_i}$ and $a^{-d_i}$, leaving their pairing unchanged).
The complementary degrees still satisfy~\eqref{eq:admissible-degrees}.
The tensor local chain complex also gives local acyclicity for both sheaves.
\end{proof}

\section{Product expansion}
\label{sec:product-expansion}

We first recall the notion of product expansion. Consider a grid $G = \prod_{j \in J} S_j$, where each $S_j$ is a set of size $n_j = |S_j|$. We say that $\ell \subseteq G$ is a coordinate line in direction $i$ if it takes the form \[\ell = S_i \times \prod_{j \in J \setminus i} \{a_j\}\]
for some $a_j \in S_j$ for each $j\neq i$. The coordinate line has length $|\ell| = n_i$. Let $C_i\subseteq F^{S_i}$ be a linear code on $S_i$. For a coordinate line $\ell$ in direction $i$ on the grid $G$, we define $C_\ell \subseteq F^{G}$ as a copy of $C_i$ supported on $\ell \subseteq G$.
Then, we consider the linear code \[\mathcal{B}((C_j)_{j \in J}) = \left(\bigotimes_{j\in J} C_j^\perp \right)^\perp = \sum_\ell C_\ell,\]
where the summation is taken over coordinate lines in all directions $i \in J$. We sometimes denote it by $\mathcal{B}$ when $(C_j)_{j \in J}$ is clear from context. 
When defining product expansion, we allow different code lengths of $(C_j)_{j\in J}$, and each coordinate line will be weighted by its length.

\begin{definition}[Product expansion]\label{def:product-expansion}
A tuple of codes $(C_i)_{i\in J}$ is \emph{$\varrho$-product expanding} if every
codeword $x\in \mathcal{B} = \sum_\ell C_\ell$ has a decomposition
$x=\sum_\ell a_\ell$, with $a_\ell\in C_\ell$, such that
\begin{equation}
 \varrho\sum_{\ell:a_\ell\ne0}|\ell|\le |x|.
 \label{eq:product-expansion}
\end{equation}
\end{definition}
Throughout, $\varrho$ denotes a product-expansion constant and $\rho$
denotes the soundness of Definition~\ref{def:normalized-soundness};
the two are related only through Section~\ref{sec:realization}.
This is the weighted-length form of
\cite[Definition~1]{KalachevPanteleev2025}. For equal lengths,
it concides with DLV's product expansion~\cite[Definition~5.12]{DinurLinVidick2024} up to scaling,
which refers to cocycle expansion in degree $|J|-1$.
Permuting and rescaling coordinates by nonzero factors independently
in each direction preserves product expansion.

A useful way to show product expansion is to identify inner generated sets.

\begin{definition}[Inner generation,
{\cite[Definition~5]{KalachevPanteleev2025}}]
\label{def:inner-generated}
A set $Z\subseteq G$ is \emph{inner-generated} if every codeword of
$\cB$ supported in $Z$ is a sum of codewords from $C_\ell$ with
$\ell\subseteq Z$.
\end{definition}

\begin{lemma}[Product expansion,
adapted from {\cite[proof of Lemma~19]{KalachevPanteleev2025}}]
\label{lem:inner-to-product}
Let $K\ge1$. Suppose that the support of every nonzero
$x\in\cB$ is contained in an inner-generated set $Z$ with
$|Z|\le K|\supp x|$. Then $(C_i)_{i\in J}$ is
$\frac1{|J|K}$-product expanding.
\end{lemma}
\begin{proof}
Write $x=\sum_{\ell\subseteq Z}a_\ell$ with $a_\ell\in C_\ell$.
Lines in each direction are pairwise disjoint, so their total
length inside $Z$ is at most $|Z|$. Summing over directions gives
\[
 \sum_{\ell:a_\ell\ne0}|\ell|
 \le |J|\,|Z|\le |J| K |x|,
\]
as required by Definition \ref{def:product-expansion}. The zero codeword has
the empty decomposition.
\end{proof}

\subsection{Product expansion for Reed-Solomon codes}
\label{sec:local-chain}
This section proves $\varrho$-product expansion for the Reed-Solomon
local codes of Proposition~\ref{prop:rs-sheaf}, with a constant $\varrho$
depending only on the number of directions and a fixed lower
bound on the relative redundancies. 
In particular, the bound is independent of
the code lengths, allowing us to meet the
spectral requirement of Corollary~\ref{thm:sheaf-criterion} by choosing a sufficiently large $A$ in our construction (and applying Proposition~\ref{prop:cubical-input}) while keeping a fixed ratio $\overline q/\underline q\le4^{r-1}$ in the directional degrees.

\begin{theorem}[Reed-Solomon product expansion on norm-one grids]
\label{thm:rs-product-expansion}\label{lem:rs-product-expansion}
Fix $r\ge2$ and $0<\beta\le1$. Given $A \in \mathbb N$, we set
\begin{equation}
\begin{aligned}[t]
\label{eq:field-size-choice}
q_i&=2^{2(A+i)},\quad n_i = q_i + 1, \quad
\Sigma_i=\{z\in\F_{q_i^2}:z^{q_i+1}=1\},
\end{aligned}
\end{equation}
for $i\in [r]$. For any set of integers $\beta n_i\le h_i \le n_i$ and characteristic-two field $F$ containing every $\Sigma_i$, we consider the Reed-Solomon code $C_i = \RS(\Sigma_i,n_i-h_i)$ over $F$. Then every nonempty subtuple of $(C_{i})_{i \in [r]}$ is $\varrho$-product expanding,
where
\[
\varrho=\frac1r\left(\frac{\beta}{K_r4^{r-1}}\right)^r
      =c_r\beta^r
\]
for some $K_r\ge2$ that depends only on $r$.
In particular, the bound is independent of $A$, $h_i$, $F$.

\end{theorem}

To prove the theorem, we use the inner-generation approach
of~\cite{KalachevPanteleev2025}.
For each codeword, we enlarge its support to a polynomial hull and use interpolation to show that the hull is not much larger
than the original support.
We then use Theorem~\ref{thm:inner-generation-main} and field descent to show that
the hull is inner-generated.
Finally, an application of Lemma~\ref{lem:inner-to-product} gives product expansion.
Let $F$ be a field, $S_i\subseteq F$ be finite nonempty subsets, and consider a grid $G=\prod_{i\in J}S_i$. For $d\in\mathbb Z_{\ge0}$, write
$\mathcal P_d(F)$ for the multivariate polynomials $F[(X_i)_{i\in J}]$ of
degree at most $d$ in each variable.

\begin{definition}[Polynomial hull]
\label{def:degree-hull}
For $U\subseteq G$, its \emph{degree-$d$ hull} is
\[
\operatorname{hull}_d(U)
 =\{z\in G:f(z)=0\text{ for every }f\in\mathcal P_d(F)
                  \text{ with }f|_U=0\}.
\]
\end{definition}

Thus, the hull contains $U$ and consists of the grid points that
cannot be separated from $U$ by a polynomial in $\mathcal P_d(F)$.
For instance, if $U$ contains $d+1$ points of a coordinate line, then
the whole line lies in $\operatorname{hull}_d(U)$: restricted to that
line, a polynomial in $\mathcal P_d(F)$ is a univariate polynomial of
degree at most $d$, so $d+1$ zeros force it to vanish on the whole
line. Thus, the hull fills in every coordinate line that $U$ meets in at least
$d+1$ points.

\begin{lemma}[Properties of polynomial hulls]
\label{lem:polynomial-hull}
Let $U\subseteq G$ and $Z=\operatorname{hull}_d(U)$.
\begin{enumerate}[label=(\roman*),leftmargin=*]
\item $Z$ is unchanged by extending the coefficient field $F$.
\item If $d+1\le |S_i|$ for every $i\in J$, then
\[
 |Z|\le |U|\prod_{i\in J}\frac{|S_i|}{d+1}.
\]
\item If $Z\ne G$, there is a nonzero
$E\in\mathcal P_d(\overline F)$ whose zero set on $G$ is exactly $Z$,
where $\overline F$ is an algebraic closure of $F$.
\end{enumerate}
\end{lemma}
\begin{proof}
Let $\mathrm{ev}_z\in\mathcal P_d(F)^*$ denote evaluation at $z$,
and put $W=\operatorname{span}_F\{\mathrm{ev}_u:u\in U\}$.
By definition, $z\in Z$ exactly when $\mathrm{ev}_z$ vanishes
on the common kernel of the functionals $\mathrm{ev}_u$, $u\in U$.
In finite dimension, the functionals vanishing on this common
kernel are precisely their linear combinations. Hence
\[
 z\in Z\quad\Longleftrightarrow\quad\mathrm{ev}_z\in W.
\]
Using the monomial basis of $\mathcal P_d(F)$, this says that
adding the evaluation row at $z$ to the evaluation matrix on $U$
does not increase its rank. These matrices have entries in $F$,
and their ranks are unchanged by field extension. Thus membership
in $Z$ is unchanged by field extension, proving (i).

For (ii), choose $T_i\subseteq S_i$ of size $d+1$ uniformly and
independently, and put $T=\prod_iT_i$.
Products of the univariate Lagrange polynomials on the $T_i$
form a basis dual to the evaluations at $T$.
Consequently, the evaluations at $T\cap Z$ are independent and
lie in $W$, so $|T\cap Z|\le\dim W\le|U|$.
Taking expectations gives
\[
|Z|\prod_{i\in J}\frac{d+1}{|S_i|}
 =\mathbb E|T\cap Z|\le |U|.
\]

For (iii), work over $\overline F$ and let
$V=\{f\in\mathcal P_d(\overline F):f|_U=0\}$.
By (i), each $z\in G\setminus Z$ gives a proper subspace
$\{f\in V:f(z)=0\}$ of $V$.
A vector space over an infinite field is not the union of
finitely many proper subspaces. Choose $E\in V$ outside their
union. It vanishes on $Z$ and nowhere else on $G$; since
$G\setminus Z\ne\varnothing$, it is nonzero.
\end{proof}

Next, we give a sufficient condition under which the zero set $Z_E = \{ x \in G : E(x) = 0 \}$ of a polynomial $E$ is inner-generated, which we prove in Appendix~\ref{sec:algebra}.

\begin{theorem}[Inner generation of polynomial zero sets]
\label{thm:inner-generation-main}
For every $r \ge 2$, there exists a constant $K_r \ge 2$
such that the following holds for every non-empty set $J \subseteq [r]$.
\begin{itemize}[itemsep=0pt]
    \item Let $F, \Sigma_i$ be chosen as in \eqref{eq:field-size-choice}, $S_i \subseteq \Sigma_i$ be any non-empty set, and $n_i = |S_i|$ for each $i \in J$.
    \item Let $C_i = \RS(S_i,n_i-h_i)$ be a code with redundancy $h_i$ defined over $\overline{F}$ for $i \in J$.
    \item Let $E \in \overline{F}[X_i : i\in J]$ be a nonzero polynomial of degree at most $d \in \mathbb N$ in each variable.
\end{itemize}
If $h_i \ge K_r d$ for every index $i\in J$, then every codeword of
$\mathcal{B}((C_j)_{j\in J})$ supported on the zero set $Z_E$
is a sum of line-codewords whose associated
coordinate lines are contained in $Z_E$.
\end{theorem}

We now prove the product expansion of our Reed-Solomon construction (Theorem~\ref{thm:rs-product-expansion}).

\begin{proof}[Proof of Theorem~\ref{thm:rs-product-expansion}]
Fix a nonempty $J\subseteq[r]$, write
$G_J=\prod_{i\in J}\Sigma_i$ and $\cB_J=\mathcal{B}((C_j)_{j\in J}) = \sum_\ell C_\ell$, and
take a nonzero codeword $x\in\cB_J$. Notice that $n_i = q_i + 1 = |\Sigma_i|$. Put
$$U=\supp x,\qquad H=\min_{i\in J}h_i,\qquad
 d=\lfloor H/K_r\rfloor,\qquad Z=\operatorname{hull}_d(U).$$   By Lemma \ref{lem:inner-to-product}, it suffices to show that $Z$ is
inner-generated and $|Z|\le K|U|$, where
$K=(K_r4^{r-1}/\beta)^r$.

We first show that $|Z|\le K|U|$. Since $d\le H/K_r\le h_i/2\le n_i/2$ and $n_i\ge2$, we have $d+1\le n_i$. The formula for $n_i$ gives
$n_i/n_j\le4^{r-1}$ for all $i,j\in[r]$.
Thus $H\ge\beta\min_{j\in J}n_j$ and $d+1>H/K_r$ imply
\[
 \frac{n_i}{d+1}<\frac{K_r n_i}{H}
 \le\frac{K_r4^{r-1}}{\beta}\qquad(i\in J).
\]
Using $|J|\le r$, lemma \ref{lem:polynomial-hull}(ii) now gives
$|Z|\le K|U|$.

We next show that $Z$ is inner-generated over $\overline F$. If $Z=G_J$, it is inner-generated by definition.
This includes $d=0$, because $U\ne\varnothing$ and the only
constant polynomial vanishing on $U$ is zero.
Otherwise $d\ge1$, and Lemma \ref{lem:polynomial-hull} provides
a nonzero $E\in\mathcal P_d(\overline F)$ whose grid zero set is $Z$.
Since $H\ge K_rd$, Theorem \ref{thm:inner-generation-main} makes $Z$
inner-generated over $\overline F$.

Finally, we show that $Z$ is inner-generated over $F$. Consider the subspaces
\[
 L=\sum_{\ell\subseteq Z}C_\ell
 \subseteq
 V=\{y\in\cB_J:\supp y\subseteq Z\}
\]
over $F$. A matrix of generators defines $L$. Parity checks for
$\cB_J$, together with the conditions of vanishing outside $Z$,
define $V$ as a kernel. All these matrices are over $F$, so field
extension preserves their ranks and hence the dimensions of
$L$ and $V$. The equality over $\overline F$ therefore implies $L=V$
over $F$. Thus $Z$ is inner-generated over $F$.
By Lemma \ref{lem:inner-to-product}, the tuple indexed by $J$ is
$1/(|J|K)$-product expanding, and hence also
$1/(rK)=c_r\beta^r$-product expanding.
The same argument applies with $F$ replaced by any extension
field and to every word over that field, with the same constant.
\end{proof}
\subsection{From product expansion to local coboundary expansion}
\label{sec:product-to-local}

We now relate product expansion to the local coboundary expansion
condition in Definition~\ref{def:local-coboundary-expansion}.
The next proposition extends
\cite[Lemma~5.13]{DinurLinVidick2024}, which treats local codes of a
common length, to codes with different directional lengths.

\begin{proposition}[Product expansion implies robustness,
adapted from {\cite[Lemma~5.13]{DinurLinVidick2024}}]
\label{prop:product-robustness}
Fix $r\ge1$, $R\ge1$ and $\varrho>0$. Let $C_i\subseteq F^{S_i}$,
$i\in[r]$, be linear codes with positive lengths $n_i=|S_i|$.
Write $\underline n=\min_i n_i$ and $\overline n=\max_i n_i$, and assume
$\overline n/\underline n\le R$. Suppose every nonempty direction subtuple
is $\varrho$-product expanding. Then there is
$\kappa=\kappa(r,R,\varrho)>0$ such that, for every nonempty $J\subseteq[r]$,
the tensor complex
\[
 K_J=\bigotimes_{i\in J}[C_i\hookrightarrow F^{S_i}]
\]
satisfies, for $0\le j<|J|$ and $y\in K_J^j$,
\begin{equation}
 |\delta_J^j y|_b\ge
 \kappa\,\underline n\,\db(y,\im\delta_J^{j-1}).
 \label{eq:product-robustness}
\end{equation}
Here $\im\delta_J^{-1}=0$, and in degree $j$ a block is indexed by
$I\subseteq J$ with $|I|=j$ together with a point of
$\prod_{i\in I}S_i$, its value lying in
$\bigotimes_{i\in J\setminus I}C_i$.
\end{proposition}
\begin{proof}
Put $N=\overline n$ and pad each $C_i$ with zeros to length $N$.
Slicing by the positions and values of new coordinates reduces
each padded line-sum to original subtuples. Thus all padded
subtuples are $(\varrho/R)$-product expanding, since $N\le Rn_i$.
Apply \cite[Lemma~5.13]{DinurLinVidick2024} to the padded codes,
all of length $N$, with product-expansion constant $\varrho/R$.
To match their notation, choose a basis of each padded code and use
its encoder as $h_i^{\mathsf T}$. This identifies their local complex
with the padded $K_J$ and preserves block weight.
Their bound applies to cochains of minimum block weight modulo
coboundaries, so applying it to such a representative of $y$ gives
the bound in distance form.
This gives coboundary expansion at least $\kappa_0N$, where
$\kappa_0>0$ depends only on $r$ and $\varrho/R$.
Projection onto the original coordinates, acting
identically on the code factors, is a cochain retraction that
does not increase block weight. Projecting padded coboundaries
therefore cannot increase their distance from an original cochain.
The same bound holds on $K_J$, and $N\ge\underline n$ proves
the claim with $\kappa=\kappa_0$.
\end{proof}

\begin{corollary}[Uniform local robustness]
\label{prop:local-expansion}\label{arb:product-expansion}
For the parameter choice~\eqref{eq:field-size-choice}, take $F$ containing
every $\F_{q_i^2}$ and fix $0<\beta\le1/2$.
If the degrees in
Proposition \ref{prop:rs-sheaf} satisfy
\[
 \beta\le\frac{d_i+1}{q_i+1}\le1-\beta\qquad(i\in[r]),
\]
then $\cF$ and $\cF^\perp$ have local coboundary expansion
$\kappa=\kappa(r,\beta)>0$, independent of $A$ and $s$.
\end{corollary}
\begin{proof}
By Proposition~\ref{prop:rs-sheaf}, the local codes are monomially
equivalent to $\RS(\Sigma_i,d_i+1)$, and their duals to
$\RS(\Sigma_i,q_i-d_i)$.
Both have redundancy at least $\beta n_i$ by the hypothesis.
Product expansion is preserved by monomial equivalence, as noted
after Definition \ref{def:product-expansion}. Hence
Theorem \ref{thm:rs-product-expansion} gives the common constant
$\varrho=c_r\beta^r$ for all nonempty direction subtuples of both tuples.
Since $\max_i n_i/\min_i n_i\le4^{r-1}$,
Proposition~\ref{prop:product-robustness}, applied to the local tensor
complexes~\eqref{eq:tree-code-local-model}, gives the claimed local
coboundary expansion with a constant depending only on $r$ and $\beta$.
\end{proof}

\section{Rate and parameter selection}
\label{sec:realization}
The local dimensions determine the degree of positive rate.

\subsection{A lower bound on the code rate}
\label{sec:rate}
Write $C=C^\bullet(\X;\cF)$ and
$P=C^\bullet(\X;\cF^\perp)$. By Lemma~\ref{lem:cell-counts},
the number of top cells is $F_r=|G|\prod_i(q_i+1)$.

\begin{proposition}[A dimension pattern for positive rate]\label{prop:rate}
Let $\cF$ be a sheaf from Proposition~\ref{prop:rs-sheaf}.
Fix $r$, $0\le k\le r$ and integers $1\le m_i\le q_i$.
At both endpoint colors, give $k$
directions local-code dimension $m_i$ and the others dimension
$q_i+1-m_i$. Put $u_i=2m_i/(q_i+1)$. If
\[
 0<u_i\le\theta_r:=\frac1{2(r+3^r)}\qquad(i\in[r]),
\]
then $\dim_F H^k(C)\ge F_r/2$ and
\[
 \frac{\dim H^k(C)}{\dim C^k}
 \ge\frac1{2\binom rk2^{r-k}}.
\]
\end{proposition}
\begin{proof}
Number the low-dimension directions first. Cell counting gives
\begin{align*}
 p(z)&=\sum_j\frac{\dim C^j}{F_r}z^j
       =\prod_{i\le k}(u_i+z)\prod_{i>k}(2-u_i+z),\\
 p^\perp(z)&=\sum_j\frac{\dim P^j}{F_r}z^j
       =\prod_{i\le k}(2-u_i+z)\prod_{i>k}(u_i+z).
\end{align*}
Write $[z^j]$ for coefficient extraction. Duality and Euler
characteristic imply
\[
 \frac{\dim H^k(C)}{F_r}\ge\prod_i(1-u_i)
 -\sum_{\substack{j<k\\j\equiv k\pmod2}}[z^j]p(z)
 -\sum_{\substack{j>k\\j\equiv k\pmod2}}[z^{r-j}]p^\perp(z).
\]
Every coefficient below degree $k$ in $p$ has a low factor.
Set $\theta=\max_i u_i\le\theta_r$. Each such coefficient is at most
$\theta\binom rj2^{r-j-1}$. The same bound applies below degree
$r-k$ in $p^\perp$. Each of the two sums is thus at most
$\theta3^r/2$, while $\prod_i(1-u_i)\ge1-r\theta$.
The displayed lower bound is at least $1-(r+3^r)\theta$.
Finally, $\dim C^k/F_r\le\binom rk2^{r-k}$.
\end{proof}

\subsection{Realization in a prescribed degree}
\label{sec:prescribed-degree}

Uniform local expansion lets us choose the geometric degree after
fixing the local dimension ratios.

\begin{theorem}[Good codes in a prescribed degree]
\label{thm:prescribed-degree}
Fix $r\ge4$ and $2\le k\le r-2$. There is a polynomial-time computable
family of connected regular colored $r$-dimensional cubical complexes
$\X$ with locally acyclic sheaves $\cF$ over a
fixed finite characteristic-two field $F$ such that
\[
 0\longrightarrow C^0\longrightarrow\cdots
 \longrightarrow C^r\longrightarrow0,
 \qquad C^\bullet=C^\bullet(\X;\cF),
\]
has $r+1$ nonzero graded terms of dimension $\Theta(|\X(0)|)$
and differentials of bounded row and column weights, with
$\cF(\omega)=F$ for every top cell $\omega$.
The degree-$k$ CSS code has parameters $[[N,K,D,\rho]]$ with
\begin{equation}
 N=\Theta(|\X(0)|),\qquad
 \frac KN=\frac{\dim_F H^k(C)}{\dim_F C^k}
 \ge\frac1{2\binom rk2^{r-k}},\qquad
 D=\Omega(N),\qquad \rho=\Omega(1),
 \label{eq:prescribed-rate}
\end{equation}
where $D$ bounds both the cosystolic and the systolic distance and
$\rho$ bounds the soundness on both sides. The implied constants may
depend on $r,k$ and the fixed local parameters, but not on $s$. These
conclusions persist after binary expansion.
\end{theorem}

\begin{proof}
Set $\theta_r=1/[2(r+3^r)]$ and choose a dyadic
$0<\vartheta\le\theta_r/4$. Use the field sizes
\eqref{eq:field-size-choice}, so $Q_i=2^{A+i}$ and $q_i=Q_i^2$.
For $A$ large enough that all $\vartheta Q_i$ are integers, choose
\begin{equation}
 m_i=\vartheta Q_i(Q_i-1)+1,\qquad
 d_i+1=\begin{cases}
 m_i,&i\le k,\\
 q_i+1-m_i,&i>k.
 \end{cases}
 \label{eq:local-dimensions}
\end{equation}
Both choices satisfy~\eqref{eq:admissible-degrees}. As $A$ grows,
$m_i/(q_i+1)=\vartheta+O_r(2^{-A})$; hence, for all sufficiently
large $A$, the dimension fractions lie between $\vartheta/2$
and $1-\vartheta/2$, and $2m_i/(q_i+1)\le\theta_r$.
Corollary~\ref{prop:local-expansion} supplies local coboundary
expansion $\kappa=\kappa(r,\vartheta/2)>0$ for both sheaves,
independently of $A$.

Corollary~\ref{thm:sheaf-criterion} gives
$\lambda_0(r,4^{r-1},\kappa)>0$. Now choose $A$ large enough for the
preceding dimension bounds and $2^{-A}\le\lambda_0$.
Proposition~\ref{prop:cubical-input} supplies the required geometry.
Fix the coefficient field and polynomial evaluations from
Proposition~\ref{prop:rs-sheaf}. Only the congruence-prime degree
$s=\deg\ell_s$ grows thereafter.
The size estimate~\eqref{eq:quotient-size} gives $|\X(0)|\to\infty$.

The geometric and local hypotheses of the criterion hold for
the sheaf and its dual. It gives both linear distances and
both cycle and cocycle expansion bounds, since $k$ and $r-k$ are in its
range. Proposition~\ref{prop:rate} gives
$\dim H^k(C)\ge |G|\prod_i(q_i+1)/2$ and
\eqref{eq:prescribed-rate}.

Cell counts give nonzero graded terms of size $\Theta(|\X(0)|)$.
Bounded incidences and the fixed field, local dimensions and restriction matrices
give bounded binary row and column weights. Explicit group operations
and polynomial evaluations give polynomial-time construction.

Apply restriction of scalars to $C=C^\bullet(\X;\cF)$.
The trace pairing identifies its binary dual complex with the
restriction of scalars of its $F$-linear dual complex and preserves
block support, so all four block parameters are unchanged.
Restriction of scalars preserves the rate and multiplies
dimensions by the fixed factor $[F:\F_2]$.
The dimensions of the binary spaces assigned to cells are bounded, and adjacent cochain dimensions
are comparable. Lemma~\ref{lem:css-parameters} gives the parameters
$[[N,K,D,\rho]]$ with
\[
 N=\dim_{\F_2}C^k,\qquad
 \frac KN=\frac{\dim_F H^k(C)}{\dim_F C^k},\qquad
 D=\Omega(N),\qquad \rho=\Omega(1),
\]
the last bound by~\eqref{eq:css-parameter-soundness}.
\end{proof}

\begin{proof}[Proof of Theorem~\ref{thm:main-original}]
Apply Theorem~\ref{thm:prescribed-degree} with $r=4$, $k=2$ and
the binary CSS construction in its proof.
Equation~\eqref{eq:quotient-size} gives
$N=\Theta(|\X(0)|)=2^{\Theta(s)}$.
\end{proof}

\begin{remark}[Finite direct sums]
\label{rem:all-degrees}
Realize the finitely many dimension patterns at common parameters
and take the direct sum of their cochain complexes.
Each summand has linear distances and
constant soundness in every internal degree. The summand indexed
by $k$ supplies positive rate in degree $k$. Cohomology commutes
with direct sums, while coordinate weights and distances to kernels
add. The resulting complex therefore gives good codes in all internal
degrees, with only constant losses and a block $F^{r-3}$ on each top cell.
\end{remark}

\appendix

\section{Arithmetic quotients and spectral bounds}
\label{sec:arithmetic}
We verify the arithmetic and spectral assertions of Proposition~\ref{prop:cubical-input}.
We keep the notation of Section~\ref{sec:finite-presentation} and
use strong approximation and the Ramanujan theorem as arithmetic inputs.

\subsection{Edge transitivity on a Bruhat-Tits tree}
\label{sec:tree-edge-transitivity}

Use the lattice model and coloring of
Section~\ref{sec:finite-presentation}, and write
$k_\nu=(k_0)_\nu$ and $\mathcal O=\mathcal O_\nu$.

\begin{lemma}
\label{lem:tree-edge-transitivity}
The group $\SL_2(k_\nu)$ preserves vertex colors and acts transitively
on the edges of $\mathcal T_\nu$ oriented from color $0$ to color $1$,
and hence on each vertex color class.
\end{lemma}
\begin{proof}
Color preservation follows from the determinant formula in
Section~\ref{sec:finite-presentation}. Given an edge with color-$0$
endpoint $[\mathcal L]$, let $\mathcal M$ be the representative of the
other endpoint with $\pi\mathcal L\subsetneq\mathcal M\subsetneq\mathcal L$.
Lift a nonzero vector of the line $\mathcal M/\pi\mathcal L$ to
$b_1\in\mathcal M$, and a complementary one to $b_2\in\mathcal L$.
Their reductions span $\mathcal L/\pi\mathcal L$, so $(b_1,b_2)$ is an
$\mathcal O$-basis of $\mathcal L$. Moreover $\mathcal M$, being
the preimage of $\langle\bar b_1\rangle$, is $\mathcal Ob_1+\pi\mathcal L$:
\[
 \mathcal L=\mathcal O b_1\oplus\mathcal O b_2,\qquad
 \mathcal M=\mathcal O b_1\oplus\mathcal O\pi b_2.
\]
Since $\nu(\det(b_1,b_2))$ is even, a common scaling by a power
of $\pi$, followed by a unit rescaling of $b_2$, makes the
determinant $1$ without changing either lattice class.
The resulting basis matrix lies in $\SL_2(k_\nu)$ and sends the
standard edge $(o_\nu,o_\nu')$ to $([\mathcal L],[\mathcal M])$.
Transitivity on each color class follows since every vertex
is incident to an edge.
\end{proof}

\subsection{Useful arithmetic tools}
\label{sec:arithmetic-inputs}
We use the following arithmetic inputs for
Section~\ref{sec:finite-presentation}. Write
$H(K)=\{a\in A(K):\Nrd(a)=1\}$ for extensions $K/k_0$.

\begin{lemma}[Strong approximation,
{\cite[Theorem~3.1]{LSV2005}}]
\label{lem:arithmetic-approximation}
Fix $j\in[r]$. At finitely many places $\nu\ne\nu_j$, prescribe
nonempty open subsets $\Omega_\nu\subseteq H((k_0)_\nu)$.
There exists $a\in H(k_0)$ lying in every prescribed $\Omega_\nu$
and integral at all remaining places other than $\nu_j$.
Here integrality is taken in a fixed integral matrix model of $H$;
it means that the matrices of $a$ and $a^{-1}$ have entries in
$\mathcal O_\nu$.
\end{lemma}

\begin{lemma}[Ramanujan bound for congruence tree quotients,
consequence of {\cite[Theorem~1.2, $d=2$]{LSV2005}}]
\label{lem:arithmetic-ramanujan}
Let $\Lambda$ be a congruence subgroup of a cocompact arithmetic
lattice in $\PGL_2((k_0)_{\nu_j})$ obtained from $A(k_0)$.
Here congruence means containing a principal congruence subgroup,
the kernel of reduction modulo a nonzero ideal in an integral model.
If $\Lambda$ acts freely on vertices and without edge inversions,
then $\Lambda\backslash\mathcal T_{\nu_j}$ is a finite connected
$(q_j+1)$-regular graph, and its adjacency eigenvalues satisfy
\[
 |\mu|\le 2\sqrt{q_j}
 \qquad\text{unless }\mu=\pm(q_j+1).
\]
\end{lemma}
This is the congruence-subgroup form of
\cite[Theorem~1.2, $d=2$]{LSV2005}, using
\cite[Theorem~3.2]{BadulescuRoche2013} for the global
Jacquet-Langlands correspondence in positive characteristic.

\subsection{The generators and their cubical relations}
\label{sec:cubical-generator-proof}
The cubical relations are
\begin{equation}
 \begin{gathered}
 S_{\nu_i}^{-1}=S_{\nu_i},\qquad
 S_{\nu_i}S_{\nu_j}=S_{\nu_j}S_{\nu_i}\quad(i\ne j),\\
 |S_{\nu_1}\cdots S_{\nu_r}|=
 \prod_{i=1}^r|S_{\nu_i}|=\prod_{i=1}^r(q_i+1).
 \end{gathered}
 \label{eq:unique-refactorization}
\end{equation}

\begin{lemma}[Vertex transitivity and cubical relations]
\label{lem:cubical-generators}
Reduction followed by taking the image line identifies
$S_{\nu_i}$ with $\PP^1(\F_{q_i})$. The group $\Gamma_\Sigma$
acts simply transitively on the vertices of $\mathcal T_\Sigma$.
The sets $S_\nu$ are inverse-closed, satisfy
\eqref{eq:unique-refactorization}, and have injective ordered
products in distinct directions.
\end{lemma}
\begin{proof}
For polynomial coefficients, the norm formula gives
\[
 \deg_y\Nrd(f+gz)=\max\{2\deg f,\,2\deg g+1\}.
\]
Thus norm $P_i$ forces $\deg f=a_i/2$ and $\deg g<a_i/2$.
Reduce the splitting $\iota_{\nu_i}$ from
Section~\ref{sec:finite-presentation} modulo $P_i$.
For a prescribed line $\ell$, requiring both matrix columns to lie in
$\ell$ gives $2a_i$ linear equations over $\F_2$ in the $2a_i$
coefficients of $f-y^{a_i/2}$ and $g$.
In the homogeneous system, the norm is divisible by $P_i$ and
has smaller degree, hence is zero; the degree formula then forces
$f=g=0$. Thus each line gives a unique normalized representative,
whose monic norm must be $P_i$.

Write $\overline{f+gz}=\varphi(f)-gz$. The conjugate of a
generator satisfies the same normalization and represents its
inverse, proving inverse closure. For generators $b,c$ in
distinct directions $\nu_i,\nu_j$, choose $c'$ in direction
$\nu_j$ with the same image line as $bc$ modulo $P_j$.
Since $\bar c'$ reduces to the adjugate of $c'$, every coefficient
of $\bar c'bc$ is divisible by $P_j$. Therefore
\[
 b'=\frac{\bar c'bc}{P_j}
\]
has polynomial coefficients, norm $P_i$, and the same monic
normalization, so $bc=c'b'$ is the required square relation.
The image line makes it unique. Successively recovering factors
gives injectivity in distinct directions and compatibility
of the cube relations.

The same operation factors any $f+gz$ with $f$ monic,
$\deg g<\deg f$, and norm a product of the $P_i$: remove a
generator with the same image line at a prime factor, or divide
by $P_i$ if the reduction is zero. The norm degree decreases
at every step. Hence
\begin{equation}
 \Gamma_\Sigma=\left\{[f+gz]:
 \begin{array}{l}
 f,g\in\F_4[y],\quad f\text{ monic},\quad\deg g<\deg f,\\
 \Nrd(f+gz)=\prod_iP_i^{e_i},\quad e_i\ge0
 \end{array}\right\}.
 \label{eq:arithmetic-membership}
\end{equation}
Multiplication of normalized generators gives the other inclusion.
An element fixing the standard vertex in every selected tree
can be scaled to have polynomial coefficients and constant norm;
the degree formula and normalization then make it $1$.
The image-line construction supplies every neighbor in every
direction, proving simple transitivity on $\mathcal T_\Sigma$.
\end{proof}

Put $t=1/y$. Equivalently, $\Gamma_\Sigma$ consists of projective
classes with integral-unit representatives at finite places outside
$\Sigma$, including $1+y$, and a representative at infinity satisfying
\begin{equation}
 f\in1+t\F_4[[t]],\qquad g\in t\F_4[[t]].
 \label{eq:infinity-normalization}
\end{equation}
Here an integral unit has integral coefficients and unit norm.
Indeed, clear scalar denominators and remove common factors in
$\F_2[y]$ from $f,g$. A prime $p\notin\Sigma$ dividing the norm
would force $v_p(\Nrd(f+gz))=2m>0$ and $p^m\mid f,g$, by the
local unit condition, a contradiction. Thus the norm has
only the factors $P_i$; the degree formula and
\eqref{eq:infinity-normalization} give $\deg g<\deg f$ and $f$ monic.
The converse follows by scaling a representative
in~\eqref{eq:arithmetic-membership} by $t^{\deg f}$ at infinity.

\subsection{Congruence images and injectivity}
\label{sec:quotient-injectivity}
Recall the color subgroup from~\eqref{eq:type-norm-one}.
Equivalently, $\Gamma_\Sigma^0$ consists of norm-one quaternions
with $f,g\in\F_4[y,P_1^{-1},\ldots,P_r^{-1}]$ satisfying
\eqref{eq:infinity-normalization}. These local conditions are open.
We apply Lemma~\ref{lem:arithmetic-approximation}, leaving
$\nu_j$ unrestricted. This gives elements of $\Gamma_\Sigma^0$
that approximate prescribed determinant-one matrices at the
other selected places and have any prescribed reduction
in $\SL_2(L_s)$.
Such reductions lift locally by lifting elementary matrices through
an integral splitting at $\ell_s$.

\begin{proposition}[Congruence images and injectivity]
\label{prop:finite-quotients}
For these levels $I_s$,
\[
 \pi_{I_s}(\Gamma_\Sigma^0)=G_s,\qquad
 \Gamma_\Sigma/\Gamma_\Sigma(I_s)
 \cong G_s\times\F_2^\Sigma.
\]
For $s>2r a_{\max}$, where $a_{\max}=\max_i a_i$,
the identification~\eqref{eq:finite-cell-definition} holds.
The deck group acts freely on vertices, and the quotient map
$p_s:\mathcal T_\Sigma\to\X_s$ is injective on vertices in every
ball of radius $\lfloor(s-1)/(2a_{\max})\rfloor$.
In particular, it is injective on each cube and preserves complete
upper stars.
\end{proposition}
\begin{proof}
For the matrix realization of reduction, put
$R_\Sigma=\F_2[y,(1+y)^{-1},P_1^{-1},\ldots,P_r^{-1}]$ and
\[
 A(R_\Sigma)=(\F_4[y]\oplus\F_4[y]z)\otimes_{\F_2[y]}R_\Sigma
 \subseteq A(k_0).
\]
Each generator has norm $P_i\in R_\Sigma^\times$, so
$\Gamma_\Sigma\subseteq A(R_\Sigma)^\times/R_\Sigma^\times$.
To split the algebra modulo $I_s$, set
$E_s=L_s\otimes_{\F_2}\F_4$ and $\eta_s=(1+y)\bmod\ell_s$.
The algebra $E_s$ is either a quadratic field or $L_s\times L_s$;
in both cases its norm is onto $L_s^\times$. Choose $c_s$ with
$c_s\varphi(c_s)=\eta_s$. Then
\[
 (a+bz)v=av+bc_s\varphi(v)\qquad(v\in E_s)
\]
identifies $A(R_\Sigma)/I_sA(R_\Sigma)$ with
$\operatorname{End}_{L_s}(E_s)\simeq M_2(L_s)$ and reduced norm
with determinant.

The preceding approximation statement, with $\nu_1$ unrestricted,
gives every reduction in $\SL_2(L_s)$ from $\Gamma_\Sigma^0$.
Since $\SL_2(L_s)=\PSL_2(L_s)=\PGL_2(L_s)$ in characteristic two,
\[
 \pi_{I_s}(\Gamma_\Sigma^0)=G_s=\PSL_2(L_s).
\]
Since the color map is onto, $(\pi_{I_s},\operatorname{color})$ maps onto
$G_s\times\F_2^\Sigma$, with kernel $\Gamma_\Sigma(I_s)$.
The images of the
generators therefore connect the entire decorated complex.

For separation, multiply the representatives of a word of length
$\ell$, using $\bar b$ for an inverse, and write the product as
$f+gz$. Its coefficient degrees are at most $a_{\max}\ell$.
Scalar reduction implies
\[
 \ell_s\mid g,\qquad \ell_s\mid f-\varphi(f).
\]
Unless the word is projectively the identity, at least one of
these polynomials is nonzero, so $\ell\ge s/a_{\max}$.
The subgroup $\Gamma_\Sigma(I_s)$ is normal in
$\Gamma_\Sigma$, which acts simply transitively on vertices.
For $v=w\mathbf o$ and a nonidentity deck transformation $\gamma$,
\[
 d_{\mathcal T_\Sigma}(v,\gamma v)
 =d_{\mathcal T_\Sigma}(\mathbf o,w^{-1}\gamma w\mathbf o)
 \ge s/a_{\max},
\]
where $d_{\mathcal T_\Sigma}$ is graph distance in the one-skeleton.
The inequality holds because $w^{-1}\gamma w$ is again a nonidentity congruence element
and its shortest word length in the directional generators equals
this graph distance. In particular, the deck group acts freely on
vertices. The asserted vertex balls are injective since twice their
radius is less than $s/a_{\max}$. When $s>2ra_{\max}$, these balls
have radius at least $r$ and contain all vertices of any cube
through their center, so the quotient map is injective on each cube.
The vertex $w\mathbf o$ maps to
$(\pi_{I_s}(w),\operatorname{color}(w))$, taking a direction-$i$
step to right multiplication by $(\pi_{I_s}(b),e_i)$.
The product cubes therefore descend to the decorated cubes
of~\eqref{eq:finite-cell-definition}. This proves
\eqref{eq:finite-cell-definition}. The upper-star assertion now
follows from the lifting argument in
Proposition~\ref{prop:tree-code-descent}.
\end{proof}

\subsection{The spectral assertion in each direction}
\label{sec:spectral-certificate}

We apply Lemma~\ref{lem:arithmetic-ramanujan} after identifying
the parallel-face graph as a tree quotient. The proof below checks
that the level and transverse stabilizer conditions are congruence
conditions and handles the color-preserving subgroup.

\begin{proof}[Proof of the spectral assertion in
Proposition~\ref{prop:cubical-input}]
Fix occupied directions $I$ and a free direction $j\notin I$.
A transverse face $\tau\subseteq\prod_{i\ne j}\mathcal T_{\nu_i}$ consists of
an edge in each direction in $I$ and a vertex in every other
direction. Orient each edge from color $0$ to color $1$.
Determinant-one matrices act transitively on transverse faces of fixed
colors. Apply Lemma~\ref{lem:arithmetic-approximation} with $\nu_j$
unrestricted and with reduction $1$ prescribed at $\ell_s$. Since face
stabilizers are open, a sufficiently good approximation acts on $\tau$
in the same way as the prescribed matrices. These moves are therefore
realized by $\Gamma_\Sigma(I_s)$.

Let $\Lambda$ be the subgroup of $\Gamma_\Sigma(I_s)$ fixing $\tau$ pointwise.
The parallel-face graph is therefore
\[
 \Lambda\backslash\mathcal T_{\nu_j},
\]
with vertices $[\tau\times\{v\}]$ and edges $[\tau\times e]$.
For the spectral bound, put
$\Lambda^+=\{g\in\Gamma_\Sigma:\pi_{I_s}(g)=1,\ g|_\tau=1\}$.
Since $\Lambda^+$ fixes the transverse colors,
$\Lambda=\Lambda^+\cap\Gamma_\Sigma^0$ has index at most two in $\Lambda^+$.
The level condition, \eqref{eq:infinity-normalization}, and the
transverse stabilizers are open congruence conditions on the arithmetic
group described above. The congruence displacement estimate in the proof of
Proposition~\ref{prop:finite-quotients} excludes vertex stabilizers
and edge inversions. Thus Lemma~\ref{lem:arithmetic-ramanujan}
applies to $\Lambda^+\backslash\mathcal T_{\nu_j}$.
Since $\Lambda$ is the subgroup preserving color $j$, our graph
is this quotient or its bipartite double cover, whose
spectrum only adds negatives of the original eigenvalues.
It is connected, bipartite and $(q_j+1)$-regular, proving
the spectral bound in Proposition~\ref{prop:cubical-input}.
\end{proof}

\section{Inner generation for Reed-Solomon codes}
\label{arb:sec-local}\label{sec:algebra}

Here, we prove the inner generation of polynomial zero sets (Theorem~\ref{thm:inner-generation-main}). Fix $r\ge2$, $\varnothing \neq J \subseteq [r]$, and $A \in \mathbb N$. Recall that $\overline{F}$ is an algebraically closed field of characteristic two, and for $i\in J$,
\[
 q_i=2^{2(A+i)},\qquad
 S_i \subseteq  \Sigma_i=\{x\in\overline F:x^{q_i+1}=1\},\qquad n_i = |S_i| >0, \qquad C_i=\RS(S_i,n_i-h_i),
\]
where $C_i$ is defined over $\overline F$. We keep track of the redundancies $1\le h_i\le n_i$ by writing
\[
\cB_h = \cB((C_j)_{j \in J}).
\]
$\cB_h$ is the span of all coordinate-line codewords on the grid $\prod_{i \in J} S_i$. Let $H=\min_{i \in J} h_i$.
For a polynomial $E \in \overline F[X_i : i\in J]$, let $Z_E$ be its zero set on the grid $\prod_{i \in J} S_i$, and
$U_E$ be the union of all coordinate lines $\ell$ of the grid contained in $Z_E$.

\begin{theorem}[Theorem \ref{thm:inner-generation-main}, Restated]
\label{thm:inner-generation-restated}
There exists a constant $K_r\ge2$ depending only on $r$, such that if $E\ne0$ has coordinate degrees
at most $d\ge1$ and $H\ge K_rd$, then every codeword of $\cB_h$
supported on $Z_E$ is a sum of line-codewords whose associated coordinate lines are contained in $Z_E$.
\end{theorem}

The proof proceeds in two steps: show that the word is supported on $U_E$,
and decompose it into line-codewords on a union of cylinders.
Both steps use induction on the number of directions.
In Section~\ref{sec:B1}, we give a criterion for membership in
$\mathcal{B}_h$ and explain its connection to polynomial interpolation
on sets. In Section~\ref{sec:B2}, we establish auxiliary lemmas that will be
used later to show that a codeword in $\mathcal{B}_h$ supported on
$Z_E$ has support contained in the union of coordinate lines lying
entirely in $Z_E$.
In Section~\ref{sec:B3}, we prove inner generation for a particular
class of sets. Finally, in Section~\ref{sec:B4}, we combine these
results to complete the proof of the inner generation theorem.

\subsection{Parity and interpolation}
\label{sec:B1}
Set $v_i(a)=\prod_{b\in S_i\setminus\{a\}}(a-b)^{-1}$.
Comparing the coefficient of $X^{n_i-1}$ in Lagrange interpolation
and using dimensions gives $C_i^\perp=\diag(v_i)\RS(S_i,h_i)$.
Taking tensor orthogonal complements gives the following test:
\begin{equation}
 M\in\cB_h\quad\Longleftrightarrow\quad
 \sum_x M(x)f(x)\prod_i v_i(x_i)=0
 \quad\text{for every }f\text{ with }\deg_i f<h_i.
 \label{eq:tensor-parity}
\end{equation}
In particular, multiplication by a polynomial of coordinate degrees
$e_i<h_i$ sends $\cB_h$ into $\cB_{h-e}$.

\begin{definition}[Interpolation with a bound]
    For a positive integer $w$, say that $W$ admits interpolation with bound $w$ if
every function on $W$ is the restriction of a polynomial with
$\deg_i f<w$.    
\end{definition}
Equivalently, for each $a\in W$ there is a polynomial
$L_a$ with these degree bounds such that $L_a(b)=\mathbf1_{a=b}$
for all $b\in W$. By~\eqref{eq:tensor-parity}, for $w\le\min_i n_i$ this is
equivalent to $\cB_{(w,\ldots,w)}\cap\overline F^W=0$.
Thus inner generation at redundancy $w$ implies interpolation
with bound $w$ on $Z_E\setminus U_E$. Interpolation also passes
to subsets.

\begin{lemma}[Unions and separation]\label{lem:finite-union}
If $W_\nu$ admits interpolation with bound $w_\nu$, then
$\bigcup_\nu W_\nu$ admits interpolation with bound $\sum_\nu w_\nu$.
If $z\notin W$ and $W$ admits interpolation with bound $w$, some
polynomial of coordinate degrees at most $w$ vanishes on $W$
but not at $z$.
\end{lemma}
\begin{proof}
Choose $L_a$ as above for $a\in W$. The polynomials
$1-\sum_aL_a$ and $(X_i-a_i)L_a$ separate every point outside $W$:
if all vanish at $z$, some $L_a(z)\ne0$ forces $z=a$.
For interpolation on a union, multiply the polynomials $L_a$ for sets
containing $a$ and separating polynomials for the other sets.
At least one factor has degree strictly below its bound in each
coordinate, giving the asserted sum bound.
\end{proof}

\subsection{Two gradient lemmas}
\label{sec:B2}
The following lemma relies on our choice of $\Sigma_i$ as a norm-one set. It uses the identity $x_i^{2^{a_i}}=x_i^{-1}$ on
$S_i\subseteq\Sigma_i$ to isolate one partial derivative at a time.

\begin{lemma}[Isolating a derivative]\label{arb:axis-crt}\label{lem:CRT}
There exists a constant $K''_r \ge 1$ depending only on $r$ with the following property.
Suppose no variable divides $F$, and write
$F=P((X_i^{2^{u_i}})_i)$, extracting the largest possible power
of two from the exponents of each variable. Put $u_i=0$ for
variables absent from $F$.
For each variable $X_j$ on which $F$ depends, there is a polynomial
$R_j$ with the same zero set as $F$ on the grid $\prod_i S_i$ with $\deg_iR_j\le K''_r\deg_iF$.
At every grid point $x \in \prod_i S_i$ with $F(x)=0$, we have $\partial_iR_j(x)=0$ for $i\ne j$, and
$\partial_j R_j(x)$ is non-zero if and only if
\[
 b_j=(\partial_jP)((X_i^{2^{u_i}})_i)
\]
is nonzero on $x$. The polynomial $b_j$ is nonzero, has smaller total
degree than $F$, and has no larger coordinate degrees.
\end{lemma}

\begin{proof}
Write $q_i = 2^{a_i}$ with $a_i = 2(A+i)$. Put $D=2\operatorname{lcm}(1,\ldots,r-1)$ and $K''_r=2^D$.
Choose $v_j=0$ and $1\le v_i\le D$ with
$v_i\equiv u_i-u_j\pmod D$ for $i\ne j$.
Since $\gcd(a_i,a_k)\mid2|i-k|\mid D$, the Chinese remainder
theorem (CRT) gives $e\ge0$ with $e\equiv u_i-v_i\pmod{a_i}$ for every $i$.
Set
\[
 \epsilon_i=(-1)^{(u_i-v_i-e)/a_i},\qquad
 \sigma(c)=c^{2^{-e}},\qquad
 R_j=\left(\prod_{\epsilon_i=-1}X_i^{2^{v_i}\deg_iP}\right)
             P^\sigma((X_i^{\epsilon_i2^{v_i}})_i),
\]
where $\sigma$ is inverse Frobenius and $P^\sigma$ applies it to
coefficients. The prefactor clears negative exponents. Distinct
monomials remain distinct, giving the stated degree bound.
Writing $m_j(X)=\prod_{\epsilon_i=-1}X_i^{2^{v_i}\deg_iP}$, the grid identities $x_i^{2^{a_i}}=x_i^{-1}$ give $$R_j(x)=m_j(x)\sigma(F(x)),\qquad
 \partial_jR_j(x)
=m_j(x)x_j^{\epsilon_j-1}\sigma(b_j(x))\text{ when }F(x)=0.$$ For $i\ne j$, we have $\partial_i R_j \equiv 0$ since $v_i \ge 1$ and all corresponding exponents are even. The displayed monomial factors are nonzero on the grid.
Maximal power extraction gives $\partial_jP\ne0$, and differentiation
lowers total degree after substitution as well. Unchosen directions
may be included as absent variables, so the same constant works
for all direction subsets.
\end{proof}

\begin{lemma}\label{arb:rank}\label{lem:rank}
Let $M\in\cB_h$ on an $m$-dimensional grid, and let $S=\supp M$.
If $H\ge mB$, no $m$ polynomials of coordinate degrees at most $B$
vanishing on $S$ have independent gradients at a point of $S$.
\end{lemma}
\begin{proof}
Suppose $F_1,\ldots,F_m$ have independent gradients at $z\in S$.
Telescoping coordinates gives
\begin{align*}
 F_j(X)&=\sum_{i=1}^m(X_i-z_i)A_{ji}(X),\\
 A_{ji}(X)&=
 \frac{F_j(z_1,\ldots,z_{i-1},X_i,\ldots,X_m)
       -F_j(z_1,\ldots,z_i,X_{i+1},\ldots,X_m)}{X_i-z_i}.
\end{align*}
Each quotient has degree at most $B-1$ in $X_i$ and $B$ in the
other variables. Hence $D=\det(A_{ji})$ has $\deg_iD\le mB-1<h_i$.
At $z$, $D(z)=\det(\partial_iF_j(z))\ne0$. At $x\in S\setminus\{z\}$,
the identity $A(x)(x-z)=0$ gives $D(x)=0$. Applying the parity
test~\eqref{eq:tensor-parity} to $D$ yields
\[
 0=\sum_x M(x)D(x)\prod_i v_i(x_i)
   =M(z)D(z)\prod_i v_i(z_i),
\]
a contradiction. 
\end{proof}

\subsection{Cylinder gluing}
\label{sec:cylinders}\label{sec:B3}

A cylinder lets some coordinates vary freely and restricts the others
to a finite base. The following result glues line-codewords on their union.

\begin{theorem}[Cylinder unions]
\label{arb:cylinders}\label{thm:cylinders}
For each nonempty $I\subseteq[m]$, let
$W_I\subseteq\prod_{i\in I}S_i$ admit interpolation with bound $w_I$.
Omit empty bases and put
\[
 Z=\bigcup_I\left(W_I\times\prod_{j\notin I}S_j\right),
 \qquad R=\max\{2,\sum_Iw_I\}.
\]
If $H\ge2mR$, then $\cB_h\cap\overline F^Z$ is generated by the
free-direction line codes in these cylinders.
\end{theorem}
\begin{proof}
We first prove a polynomial extension statement, using local
component ideals as in~\cite{DerksenSidman2004}.
Set $X=(\PP^1)^m$ and write
$\cO_X(t)=\boxtimes_i\cO_{\PP^1}(t_i)$.
Consider cylinders $Y_I=W_I\times(\PP^1)^{[m]\setminus I}$ and set $Y=\bigcup_I Y_I \subseteq X$. Let $\mathcal I$ be the ideal sheaf of $Y$. 

\emph{Local approximation.}
For each $z\in Y$, let $\mathcal J_z$ be the ideal sheaf of the union of cylinder
components containing $z$, i.e., 
$\bigcup_{I:\,z_I\in W_I}
\left(\{z_I\}\times(\PP^1)^{[m]\setminus I}\right)$.
Choose homogeneous coordinates $(u_i:v_i)$ such that
$u_i(z_i)=0$. The ideal sheaf $\mathcal J_z$
is generated by squarefree monomials in the $u_i$.
Since all least common multiples of these monomials are
squarefree, the Taylor resolution \cite[Section~2]{BPS98}
and sheafification gives a finite resolution of $\mathcal J_z$ by finite direct
sums of $\cO_X(-\varepsilon)$, with $\varepsilon\in\{0,1\}^m$.
After twisting by $t$ with $t_i\ge0$, each resolving term
is a finite direct sum of $\cO_X(t-\varepsilon)$,
with $t_i-\varepsilon_i\ge-1$.
These terms have no higher cohomology by
Lemma~\ref{lem:cyl-cohomology}.
Applying Lemma~\ref{lem:cyl-exact-complex}
shows that $\mathcal J_z(t)$ also has no higher cohomology.
Thus
\begin{equation}
    H^q(X,\mathcal J_z(t))=0
\qquad(q>0,\ t_i\ge0).
    \label{eq:cylinder-local-vanish}
\end{equation}
For each $z \notin Y$, set $\mathcal J_z=\cO_X$,
and the vanishing conclusion follows directly from
Lemma~\ref{lem:cyl-cohomology}.

Each $W_I$ admits interpolation with bound $w_I$,
also in affine coordinates $x_i=u_i/v_i$ with $v_i$
chosen nonzero on $S_i\cup\{z_i\}$.
Indeed, after homogenization, the corresponding evaluation
maps on $W_I$ differ only by nonzero weights.
For each $W_I$, take $L_{z_I}$ if $z_I\in W_I$,
and otherwise a polynomial vanishing on $W_I$
but not at $z_I$ (Lemma~\ref{lem:finite-union}).
Their product has degree at most $\sum_I w_I\le R$
in each variable.
Homogenizing to multidegree $R\mathbf1$ using the $v_i$,
which are nonzero at $z_i$, gives a section
$s_z\in H^0(X,\cO_X(R\mathbf1))$ nonzero at $z$
and vanishing on every cylinder component not containing $z$. Multiplication by $s_z$ factors thus through
\[
 \mathcal I(t)\hookrightarrow\mathcal J_z(t)
       \xrightarrow{\ \times s_z\ }\mathcal I(t+R\mathbf1),
\]
so its map on every higher cohomology group is zero for $t_i\ge0$ by factoring through (\ref{eq:cylinder-local-vanish}).

\emph{Global extension.}
Let $V$ be the span of the sections $s_z$ above,
which have no common zero. The ground field is infinite and
$\dim X=m$, so we may choose $s_0,\ldots,s_m\in V$ with no
common zero. By Lemma~\ref{lem:cyl-koszul}, their Koszul complex remains
exact after tensoring with $\mathcal I(t)$, giving
\[
0\longrightarrow \mathcal F_{m+1}
\longrightarrow\cdots
\longrightarrow\mathcal F_1
\xrightarrow{\ (s_0,\ldots,s_m)\ }
\mathcal F_0
\longrightarrow0,
\qquad
\mathcal F_j
=
\mathcal I(t-jR\mathbf1)^{\oplus\binom{m+1}{j}}.
\]
We prove by descending induction on $q=m,\ldots,1$ that
\[
H^q(X,\mathcal I(t))=0
\quad\text{whenever every }t_i\ge a_q,
\qquad
a_q=(2(m-q)+1)R.
\]
Fix $q$ and such a twist $t$, and assume the claim has
been proved in every degree $p$ with $q<p\le m$.
We first check that
$H^{q+j-1}\bigl(X,\mathcal I(t-jR\mathbf1)\bigr)=0$ for all $2\le j\le m+1$.
Write $p=q+j-1$. If $p>m$, this follows from $\dim X=m$. Otherwise, $p>q$ and $t_i-jR
\ge a_q-jR
= a_p+(j-2)R
\ge a_p$, 
so the vanishing follows from the induction hypothesis. In the base case $q=m$, we always have $p>m$, so it follows from $\dim X=m$.
Applying Lemma~\ref{lem:cohomological-surjectivity} gives
surjectivity of
\[
\Phi_q:
H^q(X,\mathcal F_1)
\xrightarrow{\ (s_0,\ldots,s_m)\ }
H^q(X,\mathcal F_0).
\]
On the other hand, since $t_i-R\ge a_q-R\ge0$, the local approximation argument above shows that the multiplication map
$H^q\bigl(X,\mathcal I(t-R\mathbf1)\bigr)
\xrightarrow{\ \times s_z\ }
H^q(X,\mathcal I(t))$ is zero on positive cohomology for every $s_z$, and thus also for every $s_\nu\in V$. Thus $\Phi_q$ is both zero and
surjective, giving $H^q(X,\mathcal I(t))=0$ and completing the induction. 
In particular, \(H^1(X,\mathcal I(t))=0\) for \(t_i\ge 2mR-1\), and every global section of \(\mathcal O_Y(t)\) extends to a global section of \(\mathcal O_X(t)\).

\emph{Line generation.}
Let a functional on \(\overline F^Z\) annihilate the free-direction
line codes, and divide its coefficients by the parity weights
\(\prod_i v_i(x_i)\).
By \eqref{eq:tensor-parity}, on each grid
\(\{a\}\times\prod_{i\notin I}S_i\subseteq Z\), with \(a\in W_I\),
the resulting function is the evaluation of a polynomial
of degree \(<h_i\) in each free variable \(x_i\).
Homogenizing this polynomial gives an \((h-\mathbf1)\)-twisted
section \(s_\alpha\) on
\(C_\alpha=\{a\}\times(\PP^1)^{[m]\setminus I}\),
where \(\alpha=(a,I)\).
These sections agree on pairwise intersections:
their polynomial representatives agree on the intersection
grids and \(h_i\le |S_i|\).
To glue these sections on \(Y=\bigcup_\alpha C_\alpha\),
work on sufficiently small affine neighborhoods on which
\(\mathcal O_X(h-\mathbf1)\) is trivial. The component ideals
\(I_\alpha\) are localizations of coordinate monomial ideals
in common coordinates, with \(I_\alpha=(1)\) for components
not containing the point.
Monomial ideals form a distributive lattice under sum and
intersection, and localization preserves these operations.
The generalized Chinese remainder theorem for distributive
families of ideals therefore applies: pairwise agreement modulo
\(I_\alpha+I_\beta\) gives a unique lift modulo
\(\bigcap_\alpha I_\alpha\).
These local lifts agree on overlaps by uniqueness and hence
glue to a section of \(\mathcal O_Y(h-\mathbf1)\).
Since \(h_i-1\ge 2mR-1\) for every \(i\), this section then extends to a global section of
\(\mathcal O_X(h-\mathbf1)\).
Dehomogenizing and restoring the parity weights gives an
extension of the original functional to the full grid,
which annihilates \(\cB_h\) by \eqref{eq:tensor-parity}.
Thus every functional annihilating the line-code span also
annihilates \(\cB_h\cap\overline F^Z\), proving equality by linear duality.
\end{proof}

\subsection{Proof of inner generation}
\label{sec:expansion}
\label{sec:B4}

\begin{proof}[Proof of Theorem~\ref{thm:inner-generation-main}]
Fix $r \ge 2$. We induct on the number $m = |J|$ of chosen directions, simultaneously
over all direction subsets and all $S_i\subseteq\Sigma_i$.
Write $J=\{j_1,\ldots,j_m\}$ and use the same constant $K''_r$
from Lemma~\ref{lem:CRT} for every subset $J$.
For $m=1$, a nonzero polynomial of degree $d$ has at most $d$
roots, and the parity Vandermonde matrix is injective on them
when $H\ge2d$. Take $K'_1=2$.
Suppose the result is known in fewer than $m$ directions with
a common constant $K'_{m-1}$. We choose $K'_m = m^2 2^{m+1}(K'_{m-1}+K''_r+1)$, and prove the result for $m$ directions as follows.

\paragraph{Interpolation on bases.}
Fix a nonzero polynomial \(F\) of coordinate degrees at most \(d\).
Group the flats contained in \(Z_F\) that are maximal under
freeing coordinates by their fixed-coordinate sets \(I\), and let
\(W_I\) denote their sets of base points. Thus
$Z_F=\bigcup_I(W_I\times\prod_{j \in J \setminus I}S_j)$.
Fix a nonempty set \(I\subseteq J\), and write $F$ as a polynomial \(F(X_I,X_{J \setminus I})=\sum_\beta f_\beta(X_I)X_{J \setminus I}^{\beta}\) with free variables $X_{J \setminus I}$. 
Since $|S_i|\ge h_i\ge H\ge K'_m d>d$, a flat lies in $Z_F$
exactly when all these coefficients vanish at its base.
Choose a linear combination \(F_I=\sum_\beta\lambda_\beta f_\beta\) whose zero set $Z_{F_I}$ concides with the common zero set of $\{f_\beta\}_\beta$ on the grid \(S_I = \prod_{i \in I} S_i\). Such a choice exists over the infinite field $\overline{F}$ by avoiding the finitely many evaluation hyperplanes associated with points outside the common zero set.
By maximality, we have  $W_I=Z_{F_I}\setminus U_{F_I}$.
When $|I|<m$, since $F_I$ has coordinate degrees at most $d$, the induction hypothesis gives inner generation of $Z_{F_I}$ at redundancy $K'_{m-1}d$, which then implies that $W_I=Z_{F_I}\setminus U_{F_I}$ admits interpolation with bound $K'_{m-1}d$ by \eqref{eq:tensor-parity}.

We also need the following consequence: If such an $F$ has fewer than
$m$ variables, then any set
$T\subseteq Z_E\setminus U_E$ annihilated by $F$ admits interpolation
with bound $2^mK'_{m-1}d$.
It suffices to consider non-constant $F$. Use the
preceding base families in the variables appearing in $F$,
leaving all other directions free. Fix one such family
$W_I\times \prod_{j\in J \setminus I} S_j$, where $1\le|I|<m$. The preceding argument gives delta
interpolants $L_a(X_I)$ on $W_I$, with coordinate degrees
less than $K'_{m-1}d$.
For each $a\in W_I$, let
$T_a=\{y\in \prod_{j \in J\setminus I} S_j :(a,y)\in T\}$.
If $E(a,X_{J\setminus I})\equiv0$, the whole flat $\{a\}\times \prod_{j \in J\setminus I}S_{j}$
lies in $U_E$, so $T_a=\varnothing$. Otherwise, $E(a,X_{J\setminus I})$ is nonzero and has
coordinate degrees at most $d$, and $T_a\subseteq Z_{E(a,\cdot)}\setminus U_{E(a,\cdot)}.$
Induction and \eqref{eq:tensor-parity}, together with
interpolation passing to subsets, give delta interpolants
$L'_{a,b}(X_{J\setminus I})$ on each nonempty $T_a$, again with coordinate
degrees less than $K'_{m-1}d$.
The product $L_a(X_I)L'_{a,b}(X_{J\setminus I})$ is then a delta interpolant
for $(a,b)$ on $T\cap(W_I\times \prod_{j \in J\setminus I} S_j)$. Since the two factors
use disjoint variable sets, the coordinate-degree bound
remains unchanged. These sets cover $T$, so
Lemma~\ref{lem:finite-union} combines the fewer than $2^m$
families to give interpolation with bound $2^mK'_{m-1}d$.

\paragraph{Restricting the support to $U_E$.}
Let $M\in\mathcal{B}_h$ be supported on $Z_E$, with $H=\min_i h_i\ge K'_m d$. We also have $M\in\mathcal B_H$, where a scalar subscript means the same redundancy in every direction. We want to show that $\supp M\subseteq U_E$.
For each direction $i$, write $E$ as a polynomial in $X_i$ and take a
linear combination of its coefficients that vanishes exactly on the bases of
the direction-$i$ lines contained in $Z_E$. Let $P_0$ be the product of these $m$ polynomials. Then $P_0$ has grid zero set exactly $U_E$ and coordinate degrees at most $(m-1)d$. Put $N=P_0M$. It suffices to prove $N=0$.

Suppose otherwise. Note that  $N = P_0 M \in \mathcal B_{H-(m-1)d}$ with \(H-(m-1)d\ge m(K''_r+m)d\) by our choice of $K'_m$. We will reach a contradiction using  Lemma~\ref{lem:rank}. Starting with $Q_0=1$, we will recursively construct polynomials $Q_t$ with each coordinate degree at most $td$ and $\supp(Q_t N)\ne\varnothing$.
Choose $F_t$ as a polynomial with minimum total degree among \[\{G \neq 0: G|_{\supp(Q_t N)} = 0\; \text{  and } \deg_i G \le d \text{ for all } i\} \supseteq \{E\}.\]
No variable can divide $F_t$ because of its minimality and that all grid coordinates are nonzero. By the preceding
interpolation argument, $F_t$ uses all $m$ variables, provided
$H-2(m-1)d\ge2^mK'_{m-1}d$.
Applying Lemma~\ref{lem:CRT} to $F_t$ in direction $j_{t+1}$, we obtain $R_{t+1}, b_{t+1}$ with the following properties:
\begin{itemize}
    \item $R_{t+1} \neq 0$ has coordinate degrees at most $K''_r d$ and $R_{t+1}|_{\supp (Q_t N)} = 0$.
    \item $b_{t+1} \neq 0$ has coordinate degrees at most $d$ and total degree smaller than that of $F_t$.
    \item For every $x \in \supp(Q_t N)$, $\partial_j R_{t+1}(x) \neq 0$ if and only if $j = j_{t+1}$ and $b_{t+1}(x) \neq 0$.
\end{itemize}
Set $Q_{t+1}=Q_tb_{t+1}$. Minimality of $F_t$ implies that
$b_{t+1}$ does not vanish entirely on $\supp(Q_t N)$, and therefore $\supp(Q_{t+1} N) \neq \varnothing$.
Construct $A_{t+1} = Q_t R_{t+1}$, which satisfies the following:
\begin{itemize}
    \item $A_{t+1}$ has coordinate degrees at most
$(K''_r+t)d$.
    \item $A_{t+1}|_{\supp(N)} = 0$, since $R_{t+1}|_{\supp(Q_t N)} = 0$ and $Q_t|_{\supp(N)\setminus \supp(Q_t N)} = 0$.
    \item On $\supp(Q_{t+1} N) \subseteq \supp(Q_{t} N)$, both $Q_t$ and $b_{t+1}$ are nowhere zero, and $R_{t+1} = 0$. Hence $\nabla A_{t+1}=Q_t\nabla R_{t+1}$ has exactly one
nonzero component, in direction $j_{t+1}$.
\end{itemize}
After $m$ steps the gradients of $A_1,\ldots,A_m$ are linearly independent
on $\supp(Q_{m} N)\ne\varnothing$, contradicting Lemma~\ref{lem:rank}.
Hence $N = 0$, and therefore $\supp M\subseteq U_E$.

\paragraph{Gluing the lines.}
Group the positive-dimensional cylinders in \(Z_E\) that are
maximal under freeing coordinates by their fixed-coordinate sets
\(\varnothing\ne I\subsetneq J\), and denote their base sets
by \(W_I\). Then
$U_E=\bigcup_{\varnothing\ne I\subsetneq J}
\left(W_I\times\prod_{j\in J \setminus I}S_j\right)$.
We know that \(M\) has support on \(U_E\). Each $W_I$ involved admit interpolation with bound $K'_{m-1}d$ by the
first paragraph. Since there are fewer than \(2^m\) indices $I$ and we have chosen
$H\ge K'_md\ge 2m\,2^mK'_{m-1}d$, Theorem~\ref{thm:cylinders} decomposes \(M\) into line-codewords
on complete coordinate lines contained in \(U_E \subseteq Z_E\). This completes the induction.

Finally, setting $K_r = \max(K'_1,\dots,K'_r)$ completes the proof.

\end{proof}


\section{Some cohomological facts and useful lemma}
\label{app:cylinder-tools}
The following are some standard consequences of the \cite{Stacks}
in the form needed for cylinder gluing. Let $k$ be a field,
$X=(\mathbb P^1_k)^m$, and
$\mathcal O_X(t)=\boxtimes_i\mathcal O_{\mathbb P^1_k}(t_i)$.
Here $H^q$ denotes ordinary sheaf cohomology.

\begin{lemma}[Line-bundle vanishing]
\label{lem:cyl-cohomology}
If $t_i\ge-1$ for every $i$, then
\[
 H^q(X,\mathcal O_X(t))=0\qquad(q>0).
\]
This follows from \cite[Lemmas~30.8.1 and~33.29.1,
Tags~01XT and~0BED]{Stacks}.
\end{lemma}

\begin{lemma}[Exactness of the Koszul complex]
\label{lem:cyl-koszul}
Let $L$ be an invertible sheaf on $X$, and let
$s_0,\ldots,s_N\in H^0(X,L)$ have no common zero.
Their Koszul complex is locally split exact and remains exact
after tensoring with any quasi-coherent sheaf.
This is the unit case of
\cite[Lemma~15.29.6, Tag~0663]{Stacks}.
\end{lemma}

We use the following standard dimension-shifting criterion.

\begin{lemma}[Cohomological surjectivity]
\label{lem:cyl-exact-complex}\label{lem:cohomological-surjectivity}
Let $N\ge1$ and let
\[
 0\longrightarrow\mathcal E_N\longrightarrow\cdots
 \longrightarrow\mathcal E_1\xrightarrow{\ d_1\ }\mathcal E_0
 \longrightarrow0
\]
be an exact complex of quasi-coherent sheaves on $X$.
For $q\ge0$, the map
$H^q(X,\mathcal E_1)\to H^q(X,\mathcal E_0)$ is surjective if
\[
 H^{q+j-1}(X,\mathcal E_j)=0\qquad(2\le j\le N).
\]
In particular, a sheaf admitting a finite resolution by sheaves
with no positive-degree cohomology has no positive-degree cohomology.
These are consequences of the hypercohomology spectral sequence
\cite[Lemma~13.21.3, Tag~015J]{Stacks}, placing
$\mathcal E_j$ in degree $-j$.
\end{lemma}

\begin{proof}
If $N=1$, exactness implies that $d_1$ is an isomorphism, so the
assertion is immediate. Let $N\ge2$.
Put $\mathcal B_0=\mathcal E_0$ and
$\mathcal B_j=\ker(\mathcal E_j\to\mathcal E_{j-1})$
for $1\le j\le N$, so that $\mathcal B_N=0$.
Exactness gives short exact sequences
\[
0\longrightarrow \mathcal B_j
 \longrightarrow \mathcal E_j
 \longrightarrow \mathcal B_{j-1}
 \longrightarrow 0.
\]
For $2\le j\le N$, the assumed vanishing and the long exact
cohomology sequence give injections
\[
H^{q+j-1}(X,\mathcal B_{j-1})
\hookrightarrow H^{q+j}(X,\mathcal B_j).
\]
Following these injections to $\mathcal B_N=0$ shows that
$H^{q+1}(X,\mathcal B_1)=0$.
The long exact sequence for $j=1$ now gives the desired
surjectivity. Applying this criterion in each positive degree
gives the final assertion about resolutions by acyclic sheaves.
\end{proof}

\phantomsection
\addcontentsline{toc}{section}{References}
\bibliographystyle{alphaurl}
\bibliography{reference}

@book{Serre1980,
  author = {Serre, Jean-Pierre},
  title = {Trees},
  publisher = {Springer-Verlag},
  address = {Berlin, Heidelberg},
  year = {1980},
  note = {Translated from the French by John Stillwell},
  doi = {10.1007/978-3-642-61856-7},
  url = {https://doi.org/10.1007/978-3-642-61856-7}
}

@article{AharonovEldar2015,
  author = {Aharonov, Dorit and Eldar, Lior},
  title = {Quantum locally testable codes},
  journal = {SIAM Journal on Computing},
  volume = {44}, number = {5}, pages = {1230--1262}, year = {2015},
  doi = {10.1137/140975498},
  url = {https://doi.org/10.1137/140975498}
}

@inproceedings{PanteleevKalachev2022,
  author = {Panteleev, Pavel and Kalachev, Gleb},
  title = {Asymptotically good quantum and locally testable classical {LDPC} codes},
  booktitle = {Proceedings of the 54th Annual ACM SIGACT Symposium on Theory of Computing},
  series = {STOC 2022}, pages = {375--388}, year = {2022},
  publisher = {Association for Computing Machinery},
  doi = {10.1145/3519935.3520017},
  url = {https://arxiv.org/abs/2111.03654v2}
}

@inproceedings{LeverrierZemor2022,
  author = {Leverrier, Anthony and Z{\'e}mor, Gilles},
  title = {Quantum {Tanner} codes},
  booktitle = {2022 IEEE 63rd Annual Symposium on Foundations of Computer Science (FOCS)},
  pages = {872--883},
  publisher = {IEEE},
  year = {2022},
  doi = {10.1109/FOCS54457.2022.00117},
  eprint = {2202.13641v3},
  archivePrefix = {arXiv},
  primaryClass = {quant-ph},
  note = {Full version revised September 16, 2022}
}

@inproceedings{DinurLinVidick2024,
  author = {Dinur, Irit and Lin, Ting-Chun and Vidick, Thomas},
  title = {Expansion of High-Dimensional Cubical Complexes: with Application to Quantum Locally Testable Codes},
  booktitle = {2024 IEEE 65th Annual Symposium on Foundations of Computer Science (FOCS)},
  pages = {379--385},
  publisher = {IEEE},
  year = {2024},
  doi = {10.1109/FOCS61266.2024.00031},
  eprint = {2402.07476v3},
  archivePrefix = {arXiv},
  primaryClass = {quant-ph},
  note = {Full version revised September 4, 2025}
}

@misc{PanteleevKalachev2024,
  author = {Panteleev, Pavel and Kalachev, Gleb},
  title = {Maximally extendable sheaf codes},
  year = {2024}, eprint = {2403.03651}, archivePrefix = {arXiv},
  primaryClass = {cs.IT}, note = {arXiv:2403.03651},
  url = {https://arxiv.org/abs/2403.03651}
}

@inproceedings{KalachevPanteleev2025,
  author = {Kalachev, Gleb and Panteleev, Pavel},
  title = {Maximally Extendable Product Codes are Good Coboundary Expanders},
  booktitle = {2025 IEEE 66th Annual Symposium on Foundations of Computer Science (FOCS)},
  pages = {1512--1524},
  publisher = {IEEE},
  year = {2025},
  doi = {10.1109/FOCS63196.2025.00079},
  eprint = {2501.01411v2},
  archivePrefix = {arXiv},
  primaryClass = {cs.IT},
  note = {Preprint revised October 21, 2025}
}

@techreport{BafnaVyas2026,
  author = {Bafna, Mitali and Vyas, Nikhil},
  title = {Private {PCPs} from product expansion},
  institution = {Electronic Colloquium on Computational Complexity},
  number = {TR26-150}, year = {2026}, month = aug,
  url = {https://eccc.weizmann.ac.il/report/2026/150/}
}

@article{DerksenSidman2004,
  author = {Derksen, Harm and Sidman, Jessica},
  title = {{Castelnuovo--Mumford} regularity by approximation},
  journal = {Advances in Mathematics},
  volume = {188}, number = {1}, pages = {104--123}, year = {2004},
  doi = {10.1016/j.aim.2003.10.001},
  eprint = {math/0211279}, archivePrefix = {arXiv},
  url = {https://arxiv.org/abs/math/0211279}
}

@article{RSV2019,
  author = {Rungtanapirom, Nithi and Stix, Jakob and Vdovina, Alina},
  title = {Infinite series of quaternionic 1-vertex cube complexes, the doubling construction, and explicit cubical {Ramanujan} complexes},
  journal = {International Journal of Algebra and Computation},
  volume = {29},
  number = {6},
  pages = {951--1007},
  year = {2019},
  doi = {10.1142/S0218196719500371},
  eprint = {1808.03290},
  archivePrefix = {arXiv},
  primaryClass = {math.GR}
}

@article{BadulescuRoche2013,
  author = {Badulescu, Alexandru Ioan and Roche, Philippe},
  title = {Global {Jacquet--Langlands} correspondence for division algebras in characteristic {$p$}},
  journal = {International Mathematics Research Notices},
  volume = {2017},
  number = {7},
  pages = {2172--2206},
  year = {2017},
  doi = {10.1093/imrn/rnw094},
  eprint = {1302.5289v2},
  archivePrefix = {arXiv},
  primaryClass = {math.NT},
  note = {Preprint revised July 15, 2014}
}

@misc{Stacks,
  author = {{The Stacks Project Authors}},
  title = {The {Stacks Project}},
  year = {2026},
  note = {Tags 01XS, 0BED, 08BV and 015J; accessed September 7, 2026},
  url = {https://stacks.math.columbia.edu/}
}

@article{CrossEtAl2024,
  author = {Cross, Andrew and He, Zhiyang and Natarajan, Anand and Szegedy, Mario and Zhu, Guanyu},
  title = {Quantum locally testable code with constant soundness},
  journal = {Quantum}, volume = {8}, pages = {1501}, year = {2024},
  doi = {10.22331/q-2024-10-18-1501},
  url = {https://doi.org/10.22331/q-2024-10-18-1501}
}

@misc{LiEtAl2025,
  author = {Li, Yiming and Li, Zimu and Liu, Zi-Wen and Nguyen, Quynh T.},
  title = {{Poincar\'e} duality and multiplicative structures on quantum codes},
  year = {2025}, eprint = {2512.21922}, archivePrefix = {arXiv},
  primaryClass = {quant-ph}, note = {arXiv:2512.21922, version 2},
  url = {https://arxiv.org/abs/2512.21922}
}

@inproceedings{HLMRZ2025,
  author = {Hsieh, Jun-Ting and Lubotzky, Alexander and Mohanty, Sidhanth and Reiner, Assaf and Zhang, Rachel Yun},
  title = {Explicit Lossless Vertex Expanders},
  booktitle = {2025 IEEE 66th Annual Symposium on Foundations of Computer Science (FOCS)},
  pages = {894--911},
  publisher = {IEEE},
  year = {2025},
  doi = {10.1109/FOCS63196.2025.00046},
  eprint = {2504.15087v1},
  archivePrefix = {arXiv},
  primaryClass = {math.CO}
}

@article{LSV2005,
  author = {Lubotzky, Alexander and Samuels, Beth and Vishne, Uzi},
  title = {Ramanujan complexes of type {$\widetilde A_d$}},
  journal = {Israel Journal of Mathematics},
  volume = {149}, pages = {267--299}, year = {2005},
  doi = {10.1007/BF02772543},
  eprint = {math/0406208}, archivePrefix = {arXiv},
  url = {https://arxiv.org/abs/math/0406208}
}

@misc{GayJeronimo2026,
  author = {Gay, William and Granha Jeronimo, Fernando},
  title = {Asymptotically Good Quantum Locally Testable Codes},
  year = {2026}, eprint = {2609.20780}, archivePrefix = {arXiv},
  primaryClass = {quant-ph}, note = {arXiv:2609.20780v1},
  url = {https://arxiv.org/abs/2609.20780}
}

@article{GimenezRuanoSanJose2023,
  author = {Gimenez, Philippe and Ruano, Diego and San-Jos{\'e}, Rodrigo},
  title = {Entanglement-assisted quantum error-correcting codes from subfield subcodes of projective {Reed--Solomon} codes},
  journal = {Computational and Applied Mathematics},
  volume = {42}, pages = {363}, year = {2023},
  doi = {10.1007/s40314-023-02506-4},
  url = {https://doi.org/10.1007/s40314-023-02506-4}
}

@article{BPS98,
  author  = {Dave Bayer and Irena Peeva and Bernd Sturmfels},
  title   = {Monomial Resolutions},
  journal = {Mathematical Research Letters},
  volume  = {5},
  number  = {1--2},
  pages   = {31--46},
  year    = {1998}
}

@techreport{BafnaLiNguyen2026,
  author      = {Mitali Bafna and Anqi Li and Quynh T. Nguyen},
  title       = {{Good Quantum Locally Testable Codes from Product Expansion}},
  institution = {Electronic Colloquium on Computational Complexity},
  number      = {TR26-204},
  year        = {2026},
  url         = {https://eccc.weizmann.ac.il/report/2026/204/}
}

@article{TradeoffConstructionsForQLTCs,
  author       = {Adam Wills and
                  Ting{-}Chun Lin and
                  Min{-}Hsiu Hsieh},
  title        = {Tradeoff Constructions for Quantum Locally Testable Codes},
  journal      = {{IEEE} Trans. Inf. Theory},
  volume       = {71},
  number       = {1},
  pages        = {426--458},
  year         = {2025},
  url          = {https://doi.org/10.1109/TIT.2024.3503500},
  doi          = {10.1109/TIT.2024.3503500},
  bibsource    = {dblp computer science bibliography, https://dblp.org}
}

@inproceedings{QuantumCodesFromHighDimensionalManifolds,
  author       = {Matthew B. Hastings},
  editor       = {Christos H. Papadimitriou},
  title        = {Quantum Codes from High-Dimensional Manifolds},
  booktitle    = {8th Innovations in Theoretical Computer Science Conference, {ITCS}
                  2017, Berkeley, CA, USA, January 9-11, 2017},
  series       = {LIPIcs},
  volume       = {67},
  pages        = {25:1--25:26},
  publisher    = {Schloss Dagstuhl - Leibniz-Zentrum f{\"{u}}r Informatik},
  year         = {2017},
  url          = {https://doi.org/10.4230/LIPIcs.ITCS.2017.25},
  doi          = {10.4230/LIPICS.ITCS.2017.25},
  bibsource    = {dblp computer science bibliography, https://dblp.org}
}

@article{TowardsLocalTestabilityForQuantumCoding,
  author       = {Anthony Leverrier and
                  Vivien Londe and
                  Gilles Z{\'{e}}mor},
  title        = {Towards local testability for quantum coding},
  journal      = {Quantum},
  volume       = {6},
  pages        = {661},
  year         = {2022},
  url          = {https://doi.org/10.22331/q-2022-02-24-661},
  doi          = {10.22331/Q-2022-02-24-661},
  bibsource    = {dblp computer science bibliography, https://dblp.org}
}

@article{GeneralDistanceBalancingForQLTCs,
  author = {Wills, Adam and Lin, Ting-Chun and Hsieh, Min-Hsiu},
  title = {Local testability of distance-balanced quantum codes},
  journal = {npj Quantum Information},
  volume = {10},
  pages = {120},
  year = {2024},
  doi = {10.1038/s41534-024-00908-8},
  eprint = {2305.00689},
  archivePrefix = {arXiv},
  primaryClass = {quant-ph},
  note = {Preprint title: General Distance Balancing for Quantum Locally Testable Codes}
}

@misc{TransversalGatesOnAlmostGoodQuantumCodes,
  author        = {Li, Yiming and Li, Zimu and Liu, Zi-Wen},
  title         = {Transversal Non-{Clifford} Gates on Almost-Good
                   Quantum {LDPC} and Quantum Locally Testable Codes},
  year          = {2026},
  archivePrefix = {arXiv},
  eprint        = {2604.01874},
  primaryClass  = {quant-ph},
  note          = {Version 2, revised July 2026},
  url           = {https://arxiv.org/abs/2604.01874v2}
}

@misc{LiLiLiu2026,
  author = {Li, Yiming and Li, Zimu and Liu, Zi-Wen},
  title = {Transversal non-{C}lifford gates on good quantum locally testable codes},
  year = {2026}, eprint = {2609.26691}, archivePrefix = {arXiv},
  primaryClass = {quant-ph}, note = {arXiv:2609.26691v1},
  url = {https://arxiv.org/abs/2609.26691}
}
\end{document}